\documentclass[12pt, draftclsnofoot, onecolumn]{IEEEtran}

\usepackage{tabularx}
\usepackage{makecell}

\usepackage[table,xcdraw]{xcolor}
\usepackage{threeparttable}
\usepackage{algpseudocode}
\usepackage{setspace}

\usepackage{stfloats}
\usepackage{amsfonts}
\usepackage{amssymb}
\usepackage{amsthm}
\usepackage{cite}
\usepackage[cmex10]{amsmath}
\usepackage{float}
\usepackage{color}
\usepackage{stfloats,fancyhdr}
\usepackage{amsmath,bm}
\usepackage{algorithm}
\usepackage{multirow}
\usepackage{changepage}
\usepackage[normalem]{ulem}
\usepackage{balance}

\newtheorem{proposition}{Proposition}

\newtheorem{definition}{Definition}

\newtheorem{remark}{Remark}

\ifCLASSINFOpdf
\usepackage[pdftex]{graphicx}
\DeclareGraphicsExtensions{.pdf,.jpeg,.png}
\else
\usepackage[dvips]{graphicx}
\DeclareGraphicsExtensions{.eps}
\fi

\usepackage{subfigure}
\usepackage{fancybox,dashbox}
\usepackage{authblk}
\usepackage{tabularx}

\usepackage{mathtools}
\mathtoolsset{showonlyrefs=true}

\begin{document}

\title{DRL-based Resource Allocation in Remote State Estimation}

\author{Gaoyang Pang, Wanchun Liu*,~\IEEEmembership{Member,~IEEE,} Yonghui Li,~\IEEEmembership{Fellow,~IEEE,} Branka Vucetic,~\IEEEmembership{Life Fellow,~IEEE,} 
\thanks{Part of the work has been submitted to IEEE Globecom 2022 \cite{Pang2022DRL}.}
\thanks{G. Pang, W. Liu, Y. Li, and B. Vucetic are with the School of Electrical and Information Engineering, The University of Sydney, Sydney, NSW 2006, Australia (e-mail: gaoyang.pang@sydney.edu.au; wanchun.liu@sydney.edu.au; yonghui.li@sydney.edu.au; branka.vucetic@sydney.edu.au). \textit{(W. Liu is the corresponding author.)}}} 

\maketitle
\vspace{-1.5cm}
\begin{abstract}
Remote state estimation, where sensors send their measurements of distributed dynamic plants to a remote estimator over shared wireless resources, is essential for mission-critical applications of Industry 4.0. Existing algorithms on dynamic radio resource allocation for remote estimation systems assumed oversimplified wireless communications models and can only work for small-scale settings. In this work, we consider remote estimation systems with practical wireless models over the orthogonal multiple-access and non-orthogonal multiple-access schemes. We derive necessary and sufficient conditions under which remote estimation systems can be stabilized. The conditions are described in terms of the transmission power budget, channel statistics, and plants’ parameters. For each multiple-access scheme, we formulate a novel dynamic resource allocation problem as a decision-making problem for achieving the minimum overall long-term average estimation mean-square error. Both the estimation quality and the channel quality states are taken into account for decision making. We systematically investigated the problems under different multiple-access schemes with large discrete, hybrid discrete-and-continuous, and continuous action spaces, respectively. We propose novel action-space compression methods and develop advanced deep reinforcement learning algorithms to solve the problems. Numerical results show that our algorithms solve the resource allocation problems effectively and provide much better scalability than the literature.
\end{abstract}
\vspace{-0.5cm}
\vspace{-0.5cm}
\section{Introduction} \label{sec:intro}
Wireless networked control systems (WNCSs), consisting
of spatially distributed plants, sensors, machines, actuators
and controllers, play an essential role in the era of Industry
4.0~\cite{Park2018WNCS}. In particular, remote state estimators for monitoring
dynamic plant status in a real-time manner are critical in
WNCSs to enable high-quality closed-loop control. In Industry
4.0, massive wireless sensors are deployed for remote state
estimation of spatially distributed plants.
This raises fundamental questions:
\begin{enumerate}
    \item What are the primary conditions on wireless radio resources to enable stable remote estimation of all plants?
    \item How to manage the limited wireless radio resources for remote state estimation of many distributed plants to achieve the optimal estimation performance under different multiple-access schemes?
\end{enumerate}

Existing works on wireless resource allocation mainly focused on data-oriented communications, and the design targets are transmission throughput, latency, and reliability \cite{Uysal2021Semantic}. Advanced data-driven machine learning approaches, such as supervised learning and reinforcement learning, have been adopted when resource allocation problems cannot be solved effectively by conventional model-based methods \cite{Zappone2019DRLoverview}. Different from data-oriented communications, the resource allocation design in a remote estimation system should be task-oriented as the ultimate goal is for minimizing the long-term average remote estimation mean-square error (MSE) of the dynamic plant states \cite{Uysal2021Semantic}. Furthermore, it is essential for a remote estimator to determine the conditions on the key parameters of wireless resources, including the transmit power budget and the channel statistics, under which the estimation process can be stabilized, i.e., the expected estimation error does not increase unboundedly as time passes by \cite{Liu2021FSMC,Liu2021Polyanskiy}. Unstable estimation will lead to unstable control, introducing catastrophic impact in real-world systems. Related works on dynamic resource allocation of remote estimation systems are summarized as below.

\subsection{Related Work}
Many existing works on remote state estimation worked on the optimal estimation algorithm design \cite{Schenato2008WNCS} and transmission control for single-plant systems~\cite{Wu2017OMA,Huang2020MDP,Wu2020OMA}. In particular, Wu et al. \cite{Wu2017OMA} considered a transmit power control problem for minimizing the instantaneous estimation MSE. The problem was solved via convex optimization.
Huang et al. \cite{Huang2020MDP} worked on a retransmission control problem of a remote estimator with hybrid automatic repeat request protocols. In order to optimize the long-term  average estimation MSE, the problem is generally formulated into a sequential decision-making problem associated with a Markov decision process (MDP). Wu et al. \cite{Wu2020OMA} further investigated a dynamic transmission scheduling problem under the constraints of communication cost. The problem was formulated as an MDP and solved by a classical reinforcement learning approach, i.e., Q-learning algorithm. Recently, wireless remote estimation for multi-plant systems has attracted a lot of attention from the control system society during the past decade, mostly focusing on transmission scheduling over limited wireless channels. 
These works can be categorized into orthogonal and non-orthogonal multiple-access (OMA and NOMA) schemes.

\textbf{OMA-based Multi-Sensor-Multi-Channel Remote Estimation:} 
each frequency channel (i.e., a subcarrier) can only be assigned to a single sensor to avoid inter-user interference completely. 
Wu et al.~\cite{Wu2018SmallAction} investigated the sensor scheduling problem, which was formulated as a multi-dimensional state and action space MDP problem to be solved by a value-iteration algorithm.
However, only problems with small-scale setups, e.g., a two-sensor-one-channel system, can be solved effectively by the classical model-based algorithm.
Leong et al. \cite{Leong2020OMA}, Demirel et al. \cite{Demirel2018OMA}, and Yang et al. \cite{Yang2022OMA} adopted deep Q-network (DQN), a data-driven deep reinforcement learning (DRL) algorithm, to solve similar scheduling problems in a relatively larger scale.
More recently, Huang et al. \cite{liu2021deep} developed an action-space reducing method of a multi-sensor-multi-actuator scheduling problem for enhancing the training efficiency.

\textbf{NOMA-based Multi-System Remote Estimation:}  multiple sensors can transmit packets at the same channel simultaneously, and a receiver detects each of the sensor packets by processing the received superimposed signals \cite{Saito2013BasicNOMA}.
Successive Interference Cancellation (SIC) is a commonly adopted scheme for decoding multiple packets arrived simultaneously.
Pezzutto et al. worked on a channel assignment problem \cite{Pezzutto2021NOMA} and a power allocation problem \cite{Pezzutto2022NOMA} for a multi-sensor-single-channel remote estimation system with NOMA schemes. Classical dynamic programming algorithms were discussed to solve the long-term average MSE minimization problem. The follow-up work \cite{Forootani2022NOMA} adopted approximate dynamic programming for reducing the computation complexity.
Li et al. \cite{Li2019NOMA2} formulated the multi-sensor transmission power control problem into an MDP and a Markov game, and developed Q-learning-based solutions.

\textbf{Stability Conditions.} The necessary and sufficient stability condition for a single-sensor-single-channel remote estimator over a static channel (i.e., fixed packet error rate) was established in terms of the packet error rate and the plant parameters \cite{Schenato2007Stability}.
Sufficient stability conditions of multi-sensor-multi-channel remote estimators over static and binary-state Markov channels were established in~\cite{Wu2018SmallAction} and \cite{Leong2020OMA}, respectively. 
More recently, we derived the necessary and sufficient stability condition of a single system  over Markov fading channels~\cite{Liu2021FSMC}.
Our follow-up work developed the stability condition for multi-sensor remote estimators assuming that the channel conditions are identical for different sensors \cite{Liu2022Stability1}.

\subsection{Limitations and Challenges}
The main limitations of existing works on resource allocation in remote estimation systems are summarized below.

\textbf{1) Small-scale systems only.} 
The existing algorithms for computing the optimal dynamic radio resource allocation policy only work for small-scale settings \cite{Wu2018SmallAction,Leong2020OMA,Demirel2018OMA,Yang2022OMA,liu2021deep,Pezzutto2021NOMA,Pezzutto2022NOMA,Forootani2022NOMA,Li2019NOMA2}, among which \cite{Leong2020OMA} considered the largest system with six sensors sharing three channels. This is because even a small system can create large state and action spaces due to the curse of dimensionality, requiring huge computation and storage resources to find an optimal policy.
In particular, the NOMA-based remote estimation works only considered a narrow-band setup, i.e., only one frequency channel is available for multi-sensor transmissions.
Finding the optimal dynamic resource allocation policy of a many-sensor-many-channel remote estimation system is challenging. 

\textbf{2) Power allocation only or channel assignment only policy.} 
The existing radio resource allocations problems in remote estimation systems either considered sensor transmit power allocation \cite{Pezzutto2022NOMA,Li2019NOMA2} or channel assignment \cite{Wu2018SmallAction,Leong2020OMA,Demirel2018OMA,Yang2022OMA,liu2021deep,Pezzutto2021NOMA,Forootani2022NOMA}. However, for NOMA-based systems, it is ideal to jointly optimize both channel assignment and power control policies.
The joint design problems can have large and hybrid action spaces, i.e., discrete channel assignment and continuous power allocation, introducing higher computation complexity in finding the optimal policies.

\textbf{3) Resource allocation without utilizing channel state information.} Existing algorithms have used instantaneous estimation quality indicators rather than channel state information (CSI) for generating resource allocation action in each time step~\cite{Forootani2022NOMA,Pezzutto2021NOMA,Pezzutto2022NOMA}.
Considering fading channel states in practice, it is necessary to utilize CSI for resource allocation.
However, using both estimation quality indicators and CSI for decision making induces a larger state space and thus higher complexity in the resource allocation policy optimization.

\textbf{4) Oversimplified channel and inaccurate packet dropout models.} 
For ease of tractability, most existing works on remote estimation considered static channels \cite{Yang2022OMA} or simplified binary-state (i.e., on-off) Markov channels \cite{Leong2020OMA}, and assumed channel conditions of different sensors are identical \cite{Leong2020OMA,Yang2022OMA,Liu2022Stability1}.
In addition, the packet error rates (i.e., reliability) were calculated assuming either a coding-free scheme or an infinite blocklength coding scheme, and were approximated using symbol error rate \cite{Li2019NOMA2} or Shannon capacity based formulas \cite{Forootani2022NOMA,Pezzutto2021NOMA}, respectively.
However, the sensor measurements are usually coded as short packages in industrial applications. 
The resource allocation algorithms obtained based on the oversimplified channel and inaccurate packet dropout models can be far from optimal in real-world systems.

\subsection{Contributions}
We systematically investigate the radio resource allocation problems of multi-sensor-multi-channel remote estimation systems over OMA and NOMA schemes. The novel contributions are summarized as below. The comparison of key features between this work and the related works in remote state estimation is summarized in Table~\ref{tab:comparison}.

\begin{table}[t]
\footnotesize
\centering
\setlength\tabcolsep{0.5pt}
\caption{Comparison of Key Features Between Existing Works in Remote Estimation and This Work}
\vspace{-0.5cm}
\begin{tabular}{cccccc}
\hline\hline
Methods & Problems & Access Scheme & Packet Dropout Models & \begin{tabular}[c]{@{}c@{}}System Scale\\ (User \#, Channel \#)\end{tabular} & Ref. \\ \hline
Dynamic programming & Channel assignment & OMA & Binary-state Markov channel & (2, 1) & \cite{Wei2021OMA} \\
\rowcolor[HTML]{EFEFEF} 
Dynamic programming & Channel assignment & NOMA & Infinite blocklength coding scheme & (2, 1) & \cite{Pezzutto2021NOMA} \\
Dynamic programming & Power allocation & NOMA & Infinite blocklength coding scheme & (2, 1) & \cite{Pezzutto2022NOMA} \\
\rowcolor[HTML]{EFEFEF} 
Dynamic programming & Channel assignment & NOMA & Infinite blocklength coding scheme & (5, 1) & \cite{Forootani2022NOMA} \\
DQN & Channel assignment & OMA & Binary-state Markov channel & (6, 3) & \cite{Leong2020OMA} \\
\rowcolor[HTML]{EFEFEF} 
DQN & Channel assignment & OMA & Error-free & (6, 3) &  \cite{Demirel2018OMA} \\
DQN & Channel assignment & OMA & Static channel & (5, 3) & \cite{Yang2022OMA} \\
\rowcolor[HTML]{EFEFEF} 
DQN & Channel assignment & OMA & Static channel & (3, 3) & \cite{liu2021deep} \\
Q learning & Power allocation & NOMA & Coding-free scheme & (3, 1) & \cite{Li2019NOMA2} \\
\rowcolor[HTML]{EFEFEF} 
\textbf{DRL} & \textbf{Joint optimization} & \textbf{OMA\&NOMA} & \textbf{Finite blocklength coding scheme} &  \textbf{(50, 25)} & \textbf{Our work} \\ \hline\hline
\end{tabular}
\label{tab:comparison}
\vspace{-1.0cm}
\end{table}

\textbf{1) New models of remote state estimation with OMA and NOMA.} 
We introduce practical short-packet transmissions and multi-state Markov fading channels in remote estimation systems for the first time. In addition to conventional OMA and multi-sensor-single-channel NOMA (where each sensor can use at most one channel), referred to as  Scenarios 1 and 2, we propose a novel multi-sensor-multi-channel NOMA scheme, referred to as  Scenario 3, with multi-round interference rejection combining (IRC) and SIC at the remote estimator to decode multiple sensor packets.
The resource allocations in Scenarios 1 and 3 require channel assignment only and power allocation only, introducing discrete and continuous actions, respectively. Scenario 2 needs hybrid discrete and continuous actions.
We comprehensively investigate the above scenarios including the decoding and resource allocation algorithm complexities.

\textbf{2) The necessary and sufficiency conditions to stabilize the remote estimation systems.}
We develop a necessary and a sufficient stability condition of the multi-sensor-multi-channel remote estimation system. Once the sufficient condition is satisfied, there exists a resource allocation policy that can stabilize the remote estimator in any of Scenarios 1-3.
If the necessary condition does not hold, there is no resource allocation policy that can stabilize the remote estimator regardless of multiple-access schemes.
We note that the stability conditions of remote estimators under the assumption of nonidentical sensors' channel states have not been discussed in the literature.

\textbf{3) Novel problem formulation with both estimation and channel quality states.} 
We formulate resource allocation problems in Scenarios 1-3 over fading channels as  Markov decision process
(MDP) problems, which take into account both estimation quality states and channel quality states for decision making.
To the best of our knowledge, such a problem has not been investigated in the literature.
The problems with large state and action spaces can not be solved by existing model-based algorithms, especially when the system scale is large.

\textbf{4) Advanced DRL algorithms for resource allocation with large state and action spaces.} 
We develop advanced data-driven DRL algorithms that generate low-dimensional continuous virtual actions. We then propose novel action mapping schemes to map virtual actions into real actions for resource allocation in Scenarios 1-3  for effectively solving the radio resource allocation problems.
Extensive simulation results illustrate that the proposed DRL algorithm can effectively solve the dynamic resource allocation problems with a much larger scale than the literature, and provides significant performance gain compared with some benchmark policies, especially in
large systems.
Furthermore, the optimized resource allocation policy with the proposed NOMA scheme can provide a significant performance gain compared to the conventional schemes.

\textbf{Outline:} The system model of the proposed remote estimation system in Scenarios 1-3 is described in Section~\ref{sec:sys}. The radio resource allocation problem formulation and the stability analysis are presented in Section~\ref{sec:MDP}. The advanced DRL algorithms for solving the formulated problems are proposed in Section~\ref{sec:DRL}. The numerical results are demonstrated and discussed in Section~\ref{sec:simulation}, followed by a conclusion in Section~\ref{sec:conclusion}.

\textbf{Notations:} Matrices and vectors are denoted by capital and lowercase upright bold letters, e.g., $\mathbf{A}$ and $\mathbf{a}$, respectively. $\mathbb{E}\left[\cdot\right]$ is the expectation operator. $\left[\mathbf{A}\right]_{i,j}$ denotes the element at $i$-th row and $j$-th column of a matrix $\mathbf{A}.\ \rho\left(\mathbf{A}\right)$ is the spectral radius of a matrix $\mathbf{A}$, i.e., the largest absolute value of its eigenvalues. $(\cdot)^\mathrm{T}$ is the vector or matrix transpose operator. $(\cdot)^\mathrm{H}$ is the vector or matrix conjugate transpose operator. $(\cdot)^{-1}$ is the matrix inverse operator. $\operatorname{Tr}(\cdot)$ is the matrix trace operator. $\mathbb{R}$ and $\mathbb{N}$ denote the sets of real number and positive integers, respectively. $\odot$ is the element-wise product operator.

\section{System Model of Remote State Estimation} \label{sec:sys}
We consider a remote estimation system with $N$ dynamic plants, each monitored by a sensor, and a remote estimator, as shown in Fig.~\ref{fig:system_model}. The $N$ sensors transmit the measurement data to the remote estimator over $M$ frequency channels (i.e., subcarriers), where $M<N$. All sensors and the remote estimator are equipped with a single antenna. The index set of the sensors is denoted as $\mathcal{N} \triangleq\{1,2, \cdots, N\}$. The remote estimator operates a dynamic radio resource allocation policy for sensor transmit power control and channel allocation.
\begin{figure}[t]
	\centering\includegraphics[width=3.5in]{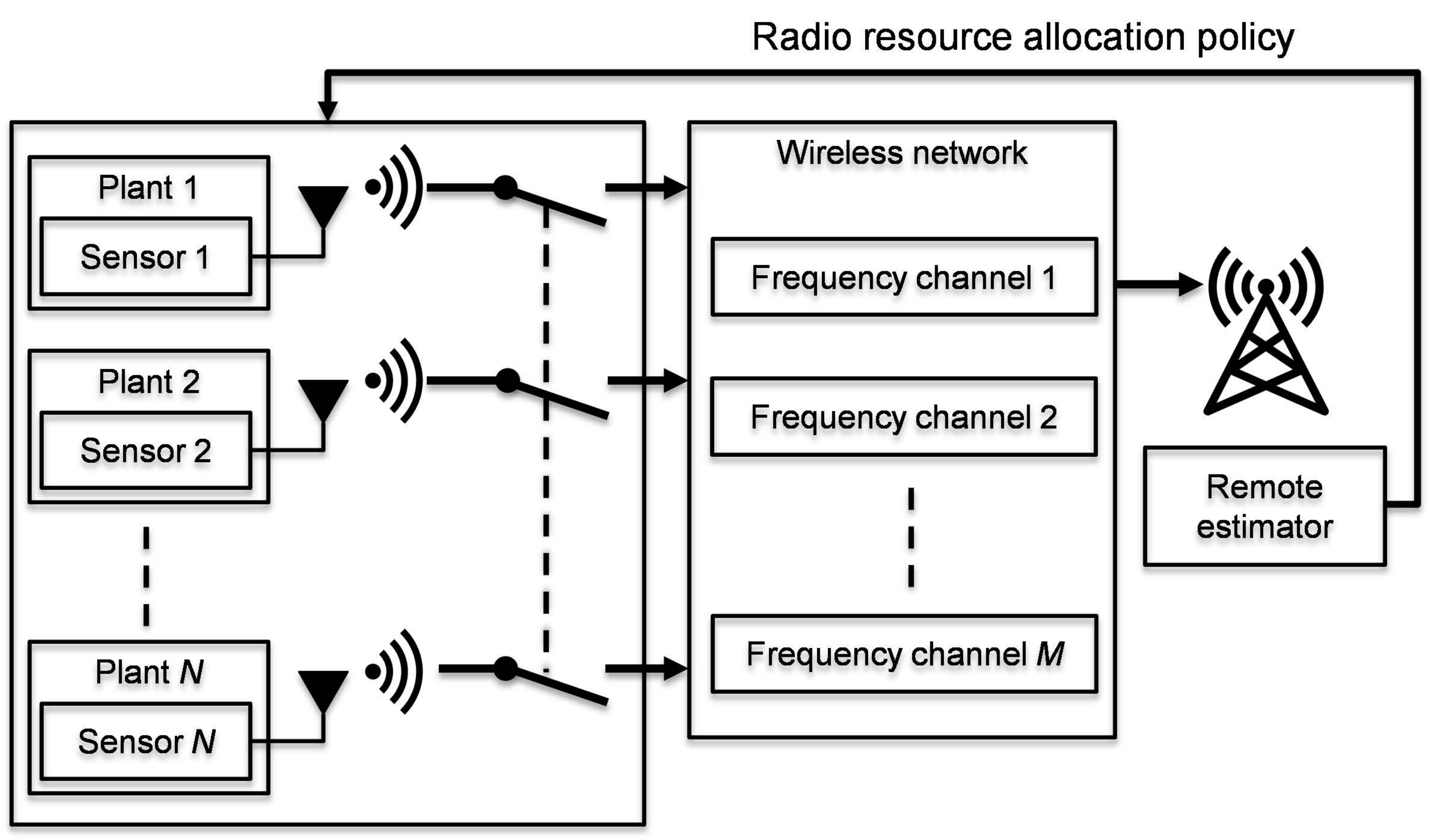}
	\vspace{-0.5cm}
	\caption{The $N$-sensor-$M$-channel remote estimation system.}
	\label{fig:system_model}
	\vspace{-0.8cm}
\end{figure}

\subsection{Local State Estimation}
Plant $n$ is modeled as the discrete-time linear time-invariant (LTI) system as \cite{Huang2020MDP} 
\begin{equation}\label{LTI}
    \begin{aligned}
    \mathbf{x}_{n}(t+1) &=\mathbf{A}_{n} \mathbf{x}_{n}(t)+\mathbf{w}_{n}(t)\\
    \mathbf{y}_{n}(t) &=\mathbf{C}_{n} \mathbf{x}_{n}(t)+\mathbf{v}_{n}(t)
    \end{aligned}
\end{equation}
where $\mathbf{x}_n(t)\in\mathbb{R}^{l_n}$ is the plant state vector; $\mathbf{A}_n\in\mathbb{R}^{l_n\times l_n}$ is the state transition matrix; $\mathbf{y}_n(t)\in\mathbb{R}^{r_n}$ is the sensor measurement vector; $\mathbf{C}_n\in\mathbb{R}^{r_n\times l_n}$ is the measurement matrix; $\mathbf{w}_n(t)\in\mathbb{R}^{l_n}$ and $\mathbf{v}_n(t)\in\mathbb{R}^{r_n}$ are the plant disturbance and sensing measurement noise vectors, respectively. These noise vectors are independent and identically distributed (i.i.d.) zero-mean Gaussian processes with covariance matrices $\mathbf{W}_n$ and $\mathbf{V}_n$, respectively.

Due to the measurement distortion and noise in \eqref{LTI}, each sensor runs a local estimator to pre-estimate the plant state $\mathbf{x}_n(t)$ based on the raw measurement $\mathbf{y}_n(t)$ before sending it to the remote estimator. We note that the KF is the optimal estimator of LTI systems in terms of the average  estimation MSE~\cite{Leong2020OMA}. The KF at sensor $n$ is given as \cite{Huang2020MDP}
\begin{subequations}
\begin{align}
\mathbf{x}_{n}^{s}\left(t\middle| t-1\right) &=\mathbf{A}_{n} \mathbf{x}_{n}^{s}(t-1)\label{KF,a}\\
\mathbf{P}_{n}^{s}\left(t\middle| t-1\right) &=\mathbf{A}_{n} \mathbf{P}_{n}^{s}(t-1) \mathbf{A}_{n}^{\mathrm{T}}+\mathbf{W}_{n}\label{KF,b}\\
\mathbf{K}_{n}(t) &=\mathbf{P}_{n}^{s}\left(t\middle| t-1\right) \mathbf{C}_{n}^{\mathrm{T}}\left(\mathbf{C}_{n} \mathbf{P}_{n}^{s}\left(t\middle| t-1\right) \mathbf{C}_{n}^{T}+\mathbf{V}_{n}\right)^{-\mathbf{1}}\label{KF,c}\\
\mathbf{x}_{n}^{s}(t) &=\mathbf{x}_{n}^{s}\left(t\middle| t-1\right)+\mathbf{K}_{n}(t)\left(\mathbf{y}_{n}(t)-\mathbf{C}_{n} \mathbf{x}_{n}^{s}\left(t\middle| t-1\right)\right)\label{KF,d}\\
\mathbf{P}_{n}^{s}(t) &=\left(\mathbf{I}_{n}-\mathbf{K}_{n}(t) \mathbf{C}_{n}\right) \mathbf{P}_{n}^{s}\left(t\middle| t-1\right)\label{KF,e}
\end{align}
\end{subequations}
where $\mathbf{I}_n$ is the $l_n\times l_n$ identity matrix; $\mathbf{x}_n^s\left(t\middle| t-1\right)$ is the prior state estimate; $\mathbf{x}_n^s(t)$ is the posterior state estimate of $\mathbf{x}_n(t)$ at time $t$, i.e., the pre-filtered measurement of $\mathbf{y}_n(t); \mathbf{K}_n(t)$ is the Kalman gain; The matrices $\mathbf{P}_n^s\left(t\middle| t-1\right)$ and $\mathbf{P}_n^s(t)$ represent the prior and posterior error covariance at the sensor at time $t$, respectively. \eqref{KF,a} and \eqref{KF,b} present the prediction steps while \eqref{KF,c}, \eqref{KF,d}, and \eqref{KF,e} correspond to the updating steps. The estimation error covariance $\mathbf{P}_n^s(t)$ is defined as
\begin{equation}\label{EstimationCovariance}
    \mathbf{P}_{n}^{s}(t) \triangleq \mathbb{E}\left[\left(\mathbf{x}_{n}^{s}(t)-\mathbf{x}_{n}(t)\right)\left(\mathbf{x}_{n}^{s}(t)-\mathbf{x}_{n}(t)\right)^{\mathrm{T}}\right].
\end{equation}

Since our work focuses on the remote state estimation, we assume that each local estimator is stable and operates in a steady state, i.e., the estimation error covariance of the local KF is a constant $\mathbf{P}_n^s(t)\triangleq{\bar{\mathbf{P}}}_n,\forall t\in\mathbb{N},n\in\mathcal{N}$ \cite{Leong2020OMA,Liu2021FSMC,Liu2022Stability1}.

\subsection{Wireless Channel Model}\label{sec:channel}
We consider finite-state time-homogeneous Markov block-fading channels \cite{Sadeghi2008FSMC}. The remote estimator has the knowledge of channel state information achieved by standard channel estimation and feedback techniques. 

Let $\mathbf{g}_{n}(t) \triangleq(g_{n, 1}(t), \ldots, g_{n, M}(t))^{\mathrm{T}}$ represent the channel power gain between sensor $n$ and the remote estimator over $M$ channels. Thus, the overall channel state information is $\mathbf{G}(t)\triangleq\left[\mathbf{g}_{1}(t),\mathbf{g}_{2}(t),\cdots,\mathbf{g}_{N}(t)\right]$. Each channel power gain has $H$ states, i.e., $\mathbf{g}_{n,m}(t) \in \mathcal{G} \triangleq\left\{h_{1}, h_{2}, \ldots, h_{H}\right\}$. The channel state vector $\mathbf{g}_{n}(t) \in \mathcal{G}^{M} \triangleq\left\{\tilde{\mathbf{g}}_{1}, \tilde{\mathbf{g}}_{2}, \ldots, \tilde{\mathbf{g}}_{H^{M}}\right\}$ is modeled as a multi-state Markov chain. Since the sensors are dislocated and have different radio propagation environment, we assume that the Markov channel states of different sensors are independent. Let $\mathbf{M}_{n} \in \mathbb{R}^{H^M \times H^M}$ denote the channel state transition matrix of sensor $n$
\begin{equation}\label{ChannelTransition}
\left[\mathbf{M}_{n}\right]_{i, j} \triangleq p_{i, j}^{h}=\operatorname{Pr}\left[\mathbf{g}_{n}(t+1)=\tilde{\mathbf{g}}_{j} \mid \mathbf{g}_{n}(t)=\tilde{\mathbf{g}}_{i}\right].
\end{equation}
We note that the channel state transition matrices of the $N$ sensors are unknown to the remote estimator. This is because the estimation of a multi-dimensional Markov chain model is computationally intensive \cite{He2017FSMC}.

We adopt the short-packet communications for real-time remote estimation \cite{Polyanskiy2010BLER}. Given the packet length $l$ (i.e., the number of symbols per packet), the number of data bits $b$, the signal-to-noise ratio (SNR) $\gamma_n$ of sensor $n$, we have the Shannon capacity $\mathcal{C}(\gamma_n)=\log_2{(1+\gamma_n)}$ and the channel dispersion $\mathcal{V}(\gamma_n)=(1-(1+\gamma_n)^{-2})(\log_2{e})^2$. Then, the decoding failure probability of sensor $n$’s packet can be approximated as \cite{Liu2021Polyanskiy}
\begin{equation}\label{BLER}
\varepsilon\left(\gamma_{n}\right) \approx \mathcal{Q}\left(\frac{\mathcal{C}\left(\gamma_{n}\right)-\frac{b}{l}}{\sqrt{\frac{\mathcal{V}\left(\gamma_{n}\right)}{l}}}\right),
\end{equation}
where $\mathcal{Q}(x)=(\frac{1}{\sqrt{2\pi}})\int_{x}^{\infty}{e^{-\frac{t^2}{2}}\text{d}t}$ is the Gaussian Q-function.

\subsection{Multiple-Access Schemes}\label{sec:MAC}
We consider three multiple-access schemes as illustrated in Fig.~\ref{fig:MAS}.
\begin{figure}[t]
	\centering
	\subfigure[]{
		\begin{minipage}[b]{0.307\linewidth}
			\includegraphics[width=1\linewidth]{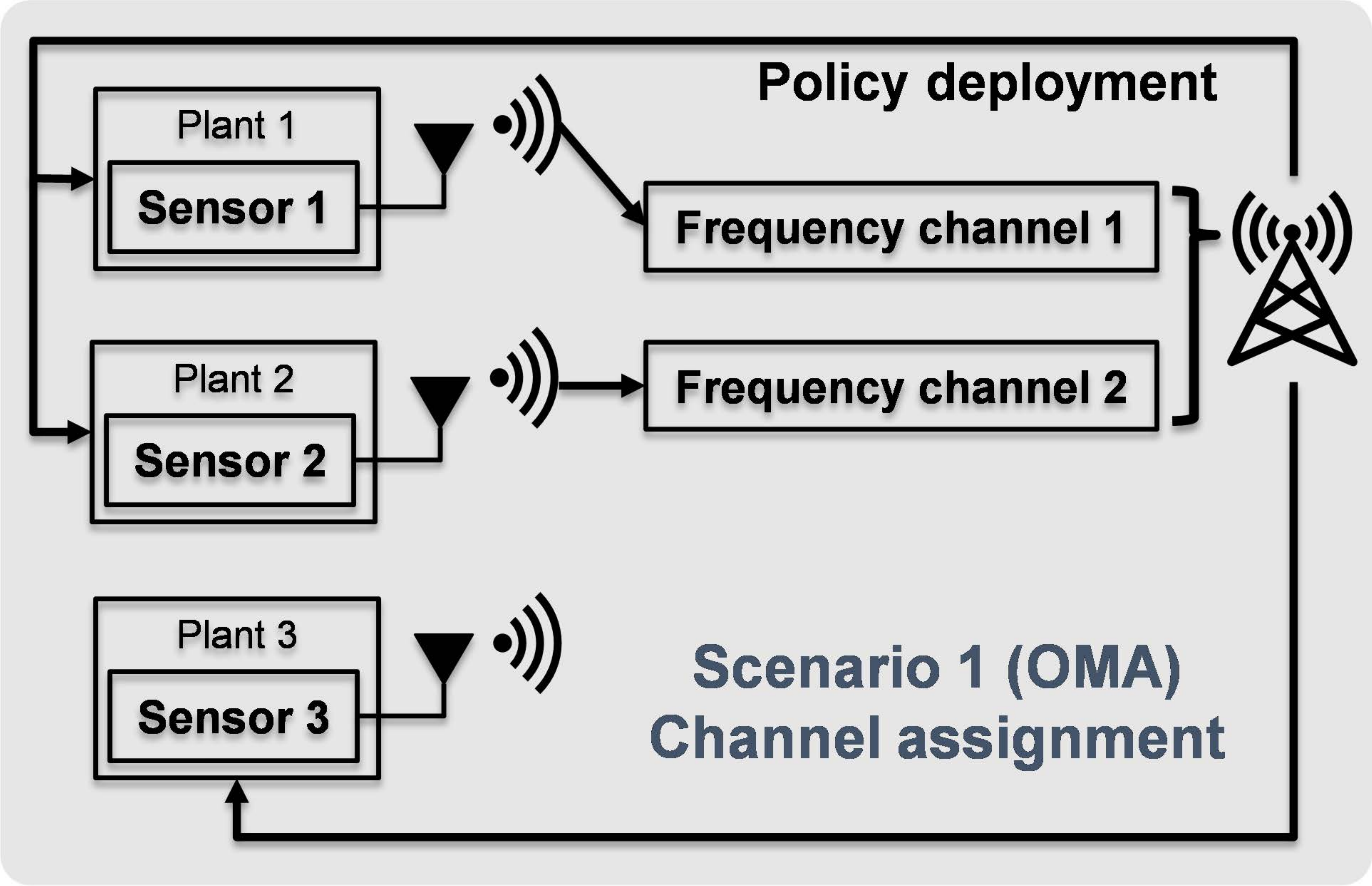}
		\end{minipage}
	}
    	\subfigure[]{
    		\begin{minipage}[b]{0.307\linewidth}
   		 	\includegraphics[width=1\linewidth]{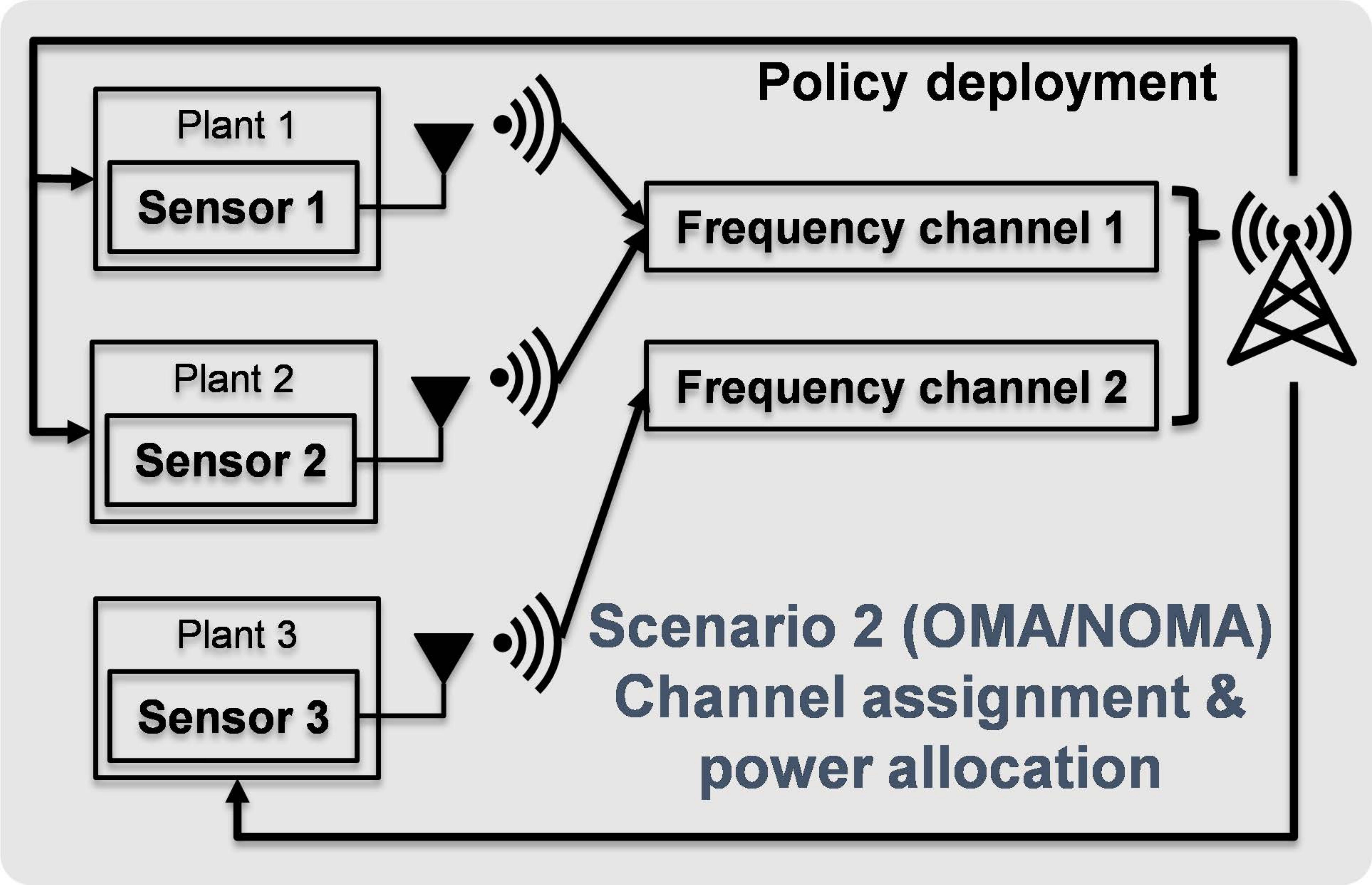}
    		\end{minipage}
    	}
    	    \subfigure[]{
        		\begin{minipage}[b]{0.307\linewidth}
       		 	\includegraphics[width=1\linewidth]{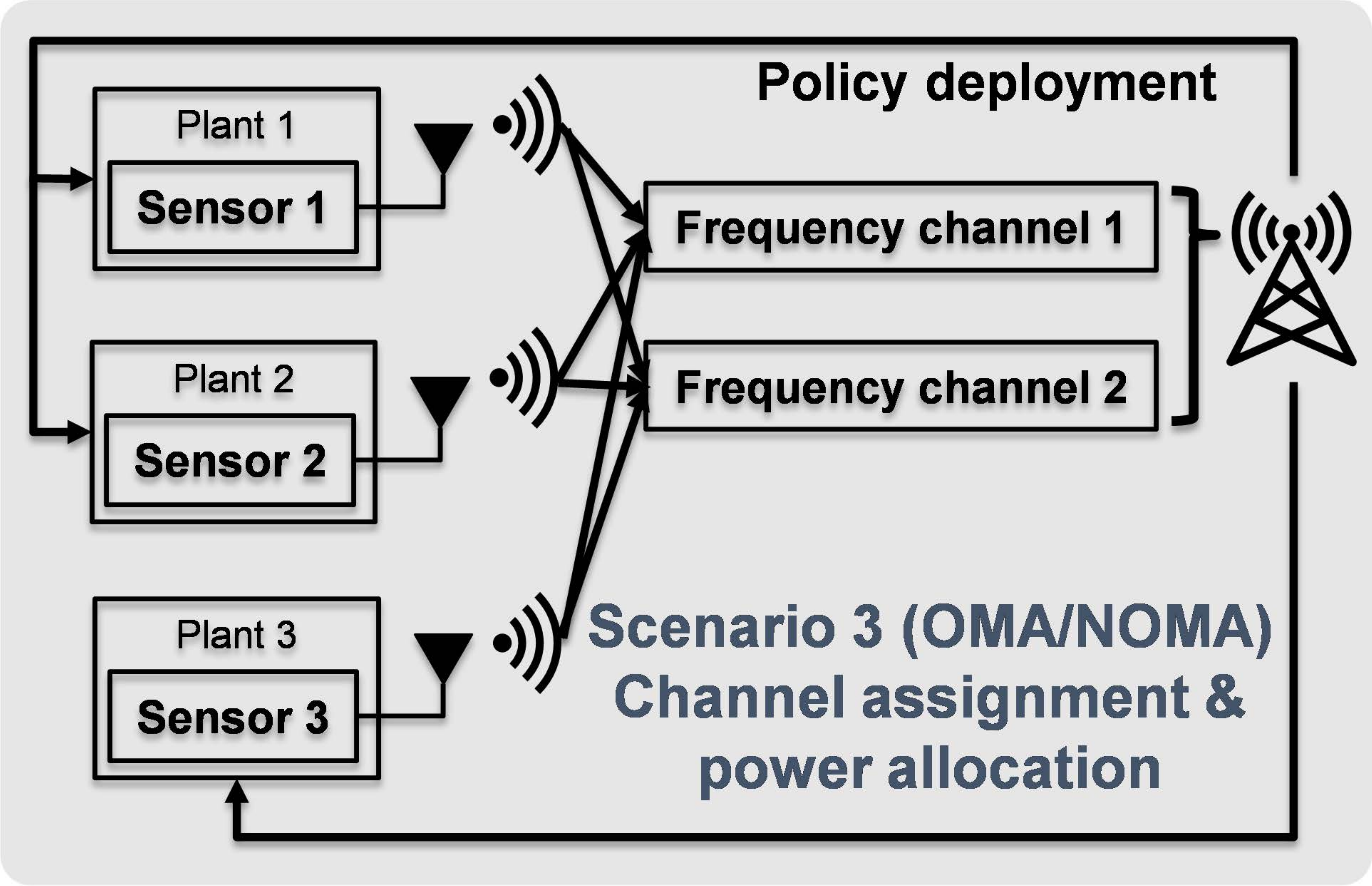}
        		\end{minipage}
        	}
    \vspace{-0.5cm}
	\caption{Three multiple-access schemes, illustrated with a $3$-sensor-$2$-channel remote estimation system. (a) Scenario 1 without IRC nor SIC. (b) Scenario 2 with SIC only. (c) Scenario 3 with IRC and SIC.}
	\vspace{-0.8cm}
	\label{fig:MAS}
\end{figure}

\subsubsection{\textbf{Scenario 1 (OMA).} Single-Sensor Single-Channel Association}
In the OMA scheme, each sensor takes at most one channel for transmission, and each channel is allocated to a single sensor. 
We denote the scheduling action of sensor $n$ at the $M$ channels as a binary vector $\mathbf{d}_n(t)\triangleq(d_{n,1}(t),\ldots,d_{n,M}(t))^{\mathrm{T}}\in\left\{0,1\right\}^M$, where $d_{n,m}(t)=1$ denotes that channel $m$ is allocated to sensor $n$ with the constraint that
\begin{equation}\label{Scen1,Constraint}
\sum_{m=1}^{M} d_{n, m}(t) \leq 1, \sum_{n=1}^{N} d_{n, m}(t)=1.
\end{equation}
Thus, the overall scheduling action of all sensors is
$\mathbf{D}(t)\triangleq\left[\mathbf{d}_{1}(t),\mathbf{d}_{2}(t),\cdots,\mathbf{d}_{N}(t)\right]$. Each scheduled sensor uses full power for transmission.

Given the transmission power budget $P_{\max}$ and the receiving noise power $\sigma^2$, the receiving SNR of sensor $n$'s signal is
\begin{equation}\label{Scen1,SINR}
\gamma_{n}(t)=\frac{P_{n}^{\mathrm{rx}}(t)}{\sigma^{2}}=\frac{\left(\mathbf{d}_{n}(t)\right)^{\mathrm{T}}\mathbf{g}_{n}(t)P_{\max }}{\sigma^{2}}.
\end{equation}
Taking \eqref{Scen1,SINR} into \eqref{BLER}, the decoding failure probability of sensor $n$’s packet can be expressed as ${\hat{\varepsilon}}_n(t)\triangleq\varepsilon(\gamma_n(t))$.

\subsubsection{\textbf{Scenario 2 (NOMA).} Multi-Sensor Single-Channel Association}
In this NOMA scheme, each sensor takes at most one frequency channel for transmission while each channel can be allocated to multiple sensors at the same time. In the scheme, at each time slot, the remote estimator needs to determine both the sensor to channel assignment and the sensor transmission power to manage the interference at the same channel.

Let $\mathbf{d}_n(t)\triangleq(d_{n,1}(t),d_{n,2}(t),\ldots,d_{n,M}(t))^{\mathrm{T}}\in\left\{0,1\right\}^M$ denote the binary channel selection action of sensor $n$ with the constraint
\begin{equation}\label{Scen2,Constraint}
\sum_{m=1}^{M} d_{n, m}(t) \leq 1.
\end{equation}
The channel allocation constraint \eqref{Scen2,Constraint} in Scenario 2 is less restrictive than \eqref{Scen1,Constraint} in Scenario 1. Let $\mathbb{P}\triangleq[0,P_{\max}]$ and $\mathbf{p}^{\mathrm{tx}}(t)\triangleq(P_1^{\mathrm{tx}}(t),P_2^{\mathrm{tx}}(t),\ldots,P_N^{\mathrm{tx}}(t))^{\mathrm{T}}\in\mathbb{P}^N$ denote the transmission power of $N$ sensors at time $t$. Then, the received signal power of sensor $n$ is
\begin{equation}
P_{n}^{\mathrm{rx}}(t)=\left(\mathbf{d}_{n}(t)\right)^{\mathrm{T}}\mathbf{g}_{n}(t)P_{n}^{\mathrm{tx}}(t).
\end{equation}

To decode the sensor packets at the same channel, the remote estimator performs SIC with the decreasing order of the received signal power of the  \cite{Ding2020SIC,Schiessl2020NOMA,Vaezi2019NOMA}, i.e., the strongest/weakest sensor signal is decoded first/last. The first sensor packet is decoded by treating all other sensor signals as interference. Once it is decoded successfully, the sensor signal can be reconstructed perfectly and thus removed from the received signal. Then, the second sensor packet will be decoded without the interference of the first one. The decoder stops once a decoding failure occurs or the last sensor packet has been successfully decoded. 

Assuming that $P_1^{\mathrm{rx}}(t)\geq P_2^{\mathrm{rx}}(t)\geq\ldots\geq P_N^{\mathrm{rx}}(t)$, the signal-to-interference-plus-noise ratio (SINR) for decoding sensor $n$’s packet is
\begin{equation}
\gamma_{n}(t)=\frac{P_{n}^{\mathrm{rx}}(t)}{\sum_{i=n+1}^{N} \left(\mathbf{d}_{n}(t)\right)^{\mathrm{T}}\mathbf{d}_{i}(t) P_{i}^{\mathrm{rx}}(t)+\sigma^{2}},
\end{equation}
where $\sigma^2$ is average power of independent additive white Gaussian noise at the remote estimator.
Thus, the decoding failure probability of sensor $n$’s packet is
\begin{equation}
\hat{\varepsilon}_{n}(t)=\!\begin{cases}
\!\varepsilon\left(\gamma_{1}(t)\right),&n=1\\
U_{n,1}+\sum_{k=2}^{n}\left(U_{n,k}\prod_{i=1}^{k-1}\left(1\!-\!U_{k,i}\right)\right),&n>1.
\end{cases}
\end{equation}
where $U_{n,n} = \varepsilon\left(\gamma_{n}(t)\right)$ and $U_{n,i} = \left(\mathbf{d}_{n}(t)\right)^{\mathrm{T}} \mathbf{d}_{i}(t) \hat{\varepsilon}_{i}(t),\forall i  \leq  n$.

\subsubsection{\textbf{Scenario 3 (NOMA).} Multi-Sensor Multi-Channel Association}
In this NOMA scheme, each sensor can broadcast its signal to multiple channels, and each channel can carry multiple sensors’ signals at the same time. Thus, Scenario 3 relaxes the constraints of channel assignment, the remote estimator solely needs to determine the transmission power of each sensor at different channels. Let $\mathbf{p}_n^{\mathrm{tx}}(t)\triangleq\left(P_{n,1}^{\mathrm{tx}}(t),P_{n,2}^{\mathrm{tx}}(t),\ldots,P_{n,M}^{\mathrm{tx}}(t)\right)^{\mathrm{T}}\in\mathbb{P}^M$ denote sensor $n$’s transmit power at the $M$ channels with the sum power constraint
$\sum_{m=1}^{M}{P_{n,m}^{\mathrm{tx}}(t)}\le P_{\max}$.
If $\mathbf{p}_n^{\mathrm{tx}}(t)=\mathbf{0}$, sensor $n$ is not scheduled for transmission at $t$. 

In contrast to Scenario 2, the remote estimator needs to jointly detect multiple sensor packets based on the received signals at the $M$ channels in Scenario 3. 
We define the $\ell$th transmitted signal (symbol) of sensor $n$'s packet in time slot $t$ as $u_n^{\mathrm{tx}}(t,\ell)\in\mathbb{C}$, where the symbol index $\ell$ is dropped for ease of notation in what follows. 
With a slight abuse of notation, let $\mathbf{g}_n(t)$ denote the channel coefficient vector of sensor $n$ in Scenario 3.
The received signal at the $M$ channels can be written~as
\begin{equation}\label{Scen3,DetectingSingal}
\mathbf{u}^{\mathrm{rx}}(t)=\sum_{n=1}^{N}\left(\sqrt{\mathbf{p}_{n}^{\mathrm{tx}}(t)} \odot \mathbf{g}_{n}(t)\right) u_{n}^{\mathrm{tx}}(t)+\mathbf{z},
\end{equation}
where $\mathbf{z}$ is the receiver noise vector with $\mathbf{Z}=\mathbb{E}\left[\mathbf{z}\mathbf{z}^H\right]=\sigma^2\mathbf{I}_M$. $\mathbf{I}_M$ denotes a $M\times M$ identity matrix. For detecting sensor $n$’s signal, \eqref{Scen3,DetectingSingal} is rewritten as
\begin{equation}
\mathbf{u}^{\mathrm{rx}}(t)=\left(\sqrt{\mathbf{p}_{n}^{\mathrm{tx}}(t)} \odot \mathbf{g}_{n}(t)\right) u_{n}^{\mathrm{tx}}(t)+\mathbf{q}_{n}(t),
\end{equation}
where $\mathbf{q}_n\left(t\right)=\sum_{i=1,i\neq n}^{N}{\left(\sqrt{\mathbf{p}_i^{\mathrm{tx}}\left(t\right)}\odot\mathbf{g}_i\left(t\right)\right)u_i^{\mathrm{tx}}\left(t\right)}+\mathbf{z}$ can be treated as the colored interference to $u_n^{\mathrm{tx}}(t)$ since the interference at different channels are correlated. The covariance matrix of $\mathbf{q}_n$~is
\begin{equation}
\mathbf{R}_{n}(t) \triangleq \mathbb{E}\left[\mathbf{q}_{n}(t)\left(\mathbf{q}_{n}(t)\right)^{\mathrm{H}}\right]=\mathbf{Z}+\sum_{i=1, i \neq n}^{N}\left(\sqrt{\mathbf{p}_{i}^{\mathrm{tx}}(t)} \odot \mathbf{g}_{i}(t)\right)\left(\sqrt{\mathbf{p}_{i}^{\mathrm{tx}}(t)} \odot \mathbf{g}_{i}(t)\right)^{\mathrm{H}}.
\end{equation}

We adopt the optimal linear combining scheme, i.e., interference rejection combining (IRC), to calculated the minimum SINR for each sensor signal \cite{Barreto2008IRC}. The IRC combiner is
\begin{equation}
\mathbf{c}_n(t)^{\mathrm{T}}\triangleq{\left(\sqrt{\mathbf{p}_n^{\mathrm{tx}}(t)}\odot\mathbf{g}_n(t)\right)}^\mathrm{H}\left(\mathbf{R}_n(t)\right)^{-1}\in\mathbb{R}^M.
\end{equation}
Therefore, the combined signal for detecting $u_n^\mathrm{tx}(t)$ is $\mathbf{c}_n(t)^{\mathrm{T}}\mathbf{u}^{\mathrm{rx}}(t)$, and the SINR is
\begin{equation}
\gamma_{n}(t)=\left(\sqrt{\mathbf{p}_{n}^{\mathrm{tx}}(t)} \odot \mathbf{g}_{n}(t)\right)^\mathrm{H} \left(\mathbf{R}_n(t)\right)^{-1}\left(\sqrt{\mathbf{p}_{n}^{\mathrm{tx}}(t)} \odot \mathbf{g}_{n}(t)\right).
\end{equation}
The remote estimator decodes the sensor packet with the largest SINR and subtracts the successfully detected signal from $\mathbf{u}^\mathrm{rx}(t)$.

Different from Scenario 2, the SINR order of the received signals can change after the SIC. Thus, we propose a \emph{multi-round IRC-SIC scheme}: in each round, the remote estimator calculates the SINR of the remaining sensor signals based on IRC and performs SIC of the signal with the highest SINR. The decoder terminates once a decoding failure occurs or all sensor packets are successfully decoded.
Assuming that sensor $n$’s packet is supposed to be decoded in the $\widetilde{n}$-th round and the remaining sensor signal index set is $\mathcal{I}_{\tilde{n}}(t)$, the SINR for detection is
\begin{equation}\label{eq:scenario3_SINR}
\tilde{\gamma}_{\tilde{n}}(t) =\left(\sqrt{\mathbf{p}_{n}^{\mathrm{tx}}(t)} \odot \mathbf{g}_{n}(t)\right)^{\mathrm{H}} \left(\tilde{\mathbf{R}}_{n}(t)\right)^{-1}\left(\sqrt{\mathbf{p}_{n}^{\mathrm{tx}}(t)} \odot \mathbf{g}_{n}(t)\right),
\end{equation}
where $\tilde{\mathbf{R}}_{n}(t)=\mathbf{Z}+\sum_{i \in \mathcal{I}_{\tilde{n}}(t)}\left(\sqrt{\mathbf{p}_{i}^{\mathrm{tx}}(t)} \odot \mathbf{g}_{i}(t)\right)\left(\sqrt{\mathbf{p}_{i}^{\mathrm{tx}}(t)} \odot \mathbf{g}_{i}(t)\right)^{\mathrm{H}}$.

Thus, the decoding failure probability of sensor $n$ is 
\begin{equation}
\hat{\varepsilon}_{n}(t)=\varepsilon\left(\tilde{\gamma}_{1}(t)\right)
+ (1-\varepsilon\left(\tilde{\gamma}_{1}(t)\right)) \varepsilon\left(\tilde{\gamma}_{2}(t)\right)+\dots +
\prod^{\tilde{n}-1}_{i=1}
(1-\varepsilon\left(\tilde{\gamma}_{i}(t)\right))
\varepsilon\left(\tilde{\gamma}_{\tilde{n}}(t)\right)
\end{equation}

\uline{\textbf{Decoding complexity}}. In Scenario 1, the remote estimator only needs to decode a single packet from each channel; Scenario 2 needs to perform SIC for multiple rounds at each channel, leading to an increased decoding complexity; Scenario 3 requires IRC for each undecoded sensor signal at each SIC round and is of higher complexity than Scenario 2. Thus, the decoding complexity increases from Scenario 1 to 3.

\subsection{Remote State Estimation}
Due to transmission scheduling and packet detection errors, sensor $n$’s packet may not be received by the remote estimator at each time slot.  Let $\zeta_n(t)=1$ denote the successful sensor $n$’s transmission at $t$. To provide real-time state estimation of all plants, the remote estimator employs a minimum mean-square error (MMSE) state estimation for each plant \cite{Leong2020OMA,Liu2021FSMC,Liu2022Stability1}
\begin{equation}\label{x_hat}
\hat{\mathbf{x}}_{n}(t)=\left\{\begin{array}{cc}
\mathbf{x}_{n}^{s}(t), & \text { if } \zeta_{n}(t)=1 \\
\mathbf{A}_{n} \widehat{\mathbf{x}}_{n}(t-1), & \text {otherwise.}
\end{array}\right.
\end{equation}
Thus, the remote estimation error covariance is
\begin{align} \label{Pn_1}
\mathbf{P}_{n}(t) & \triangleq \mathbb{E}\left[\left(\hat{\mathbf{x}}_{n}(t)-\mathbf{x}_{n}(t)\right)\left(\hat{\mathbf{x}}_{n}(t)-\mathbf{x}_{n}(t)\right)^{\mathrm{T}}\right] \\
&=\left\{\begin{array}{cc}
\mathbf{\bar{P}}_{n}, & \text { if } \zeta_{n}(t)=1 \\ \label{Pn_2}
\mathbf{A}_{n} \mathbf{P}_{n}(t-1) \mathbf{A}_{n}^{\mathrm{T}}+\mathbf{W}_{n}, & \text {otherwise.}
\end{array}\right.
\end{align}
where \eqref{Pn_2} is obtained by taking \eqref{x_hat} and \eqref{EstimationCovariance} into \eqref{Pn_1}. Recall that $\mathbf{\bar{P}}_{n}$ is the local estimation error covariance.

Now we define $\tau_n(t)$ as the age of information (AoI) of sensor $n$ at time slot $t$, representing the time interval since the last successful transmission of sensor $n$ to the remote estimator. Therefore, the AoI state of sensor $n$ has the updating rule below
\begin{equation}\label{AoI,Updating}
\tau_{n}(t)=\left\{\begin{array}{cc}
1, & \text { if } \zeta_{n}(t-1)=1 \\
\tau_{n}(t-1)+1, & \text {otherwise. }
\end{array}\right.
\end{equation}
A larger AoI indicates that the remote estimation is less accurate.
Jointly using \eqref{Pn_2} and \eqref{AoI,Updating}, the remote estimation error covariance can be concisely rewritten as a function of AoI as
\begin{equation}
\mathbf{P}_{n}(t)=f_{n}^{\tau_{n}(t)}\left(\mathbf{\bar{P}}_{n}\right),
\end{equation}
where $f_{n}^{1}(\mathbf{X}) =\mathbf{A}_{n} \mathbf{X} \mathbf{A}_{n}^{\mathrm{T}}+\mathbf{W}_{n}$ and $f_{n}^{\tau_{n}}(\mathbf{X})=f_{n}^{1}\left(f_{n}^{\tau_{n}-1}(\mathbf{X})\right)$.

To quantify the remote estimation quality of sensor $n$ at time $t$, we define the \emph{estimation cost function}, i.e., the sum estimation MSE of the plant vector state, as
\begin{equation} \label{eq:ErrorCovariance}
J_{n}(t)
\triangleq 
\mathbb{E}\left[\left(\hat{\mathbf{x}}_{n}(t)-\mathbf{x}_{n}(t)\right)^{\mathrm{T}}\left(\hat{\mathbf{x}}_{n}(t)-\mathbf{x}_{n}(t)\right)\right]
=\operatorname{Tr}\left(\mathbf{P}_{n}(t)\right)=\operatorname{Tr}\left(f_{n}^{\tau_{n}(t)}\left(\mathbf{\bar{P}}_{n}\right)\right).
\end{equation}
Thus, the AoI state and the local estimation error covariance jointly determine the remote estimation quality. A smaller $J_n(t)$ indicates that the remote state estimation is more accurate. 

Note that if $\rho(\mathbf{A}_n)\geq1$, the remote estimation cost grows up unbounded with the continuously increasing AoI \cite{Liu2021FSMC}. We focus on the stochastic stability of the remote estimator, i.e., average mean-square stability.

\begin{definition}[Mean-Square Stability \cite{Liu2021FSMC}]\label{def:stability}
	\normalfont
	The remote estimation system is average mean-square stable, if the long-term average estimation cost function is bounded, i.e.,
\begin{equation}
\lim _{T \rightarrow \infty} \mathbb{E}\left[\frac{1}{T} \sum_{t=0}^{T-1} \sum_{n=1}^{N} J_{n}(t)\right]<\infty.
\end{equation}
\end{definition}

\section{Problem Formulation} \label{sec:MDP}
For each of the three OMA and NOMA scenarios, we aim to design a 
deterministic and stationary resource allocation policy denoted as $\pi(\cdot) \in \Pi$ that generates channel allocation and power control actions at each time slot, for achieving the optimal discounted long-term average estimation MSE \cite{Leong2020OMA}, i.e.,
\begin{equation}
J^{*}=\min _{\pi(\cdot) \in \Pi} \lim _{T \rightarrow \infty} \mathbb{E}\left[\sum_{t=0}^{T-1} \sum_{n=1}^{N} \beta^{t} J_{n}(t)\right],
\end{equation}
where $\beta\in(0,1)$ is the discount factor.

This is a sequential decision-making problem and is commonly formulated into an MDP. In the sequel, we will present the MDP formulation for each scenario and discuss the stability condition under which the MDPs has a solution. Advanced DRL algorithms will be proposed in Section~\ref{sec:DRL} for solving the MDP problems.

\subsection{MDP Formulation} \label{subsec:MDP}
In general, the MDP takes the observable states $\mathbf{s}(t)$ of the remote estimation system, including both the channel state and the estimation quality state (i.e., the AoI state) and generates the resource allocation action $\mathbf{a}(t)$ at each time by following a policy. The policy is a mapping between the state and the action as $\mathbf{a}(t)=\pi(\mathbf{s}(t))$, where $\pi(\cdot)\in\Pi$. The Markovian property holds clearly due to the Markov channel modeling and the AoI state updating rule. 

\textit{\textbf{States:}} Given the AoI state $\boldsymbol{\tau}(t)=(\tau_1(t),\tau_2(t),\ldots,\tau_N(t))\in\mathbb{N}^N$, and the channel state matrix $\mathbf{G}(t)\in\mathcal{G}^{M\times N}$, the state of the MDP is defined as $\mathbf{s}\left(t\right)\triangleq\left\{\mathbf{G}(t),\boldsymbol{\tau}(t)\right\}\in\mathcal{G}^{M\times N}\times\mathbb{N}^N$.

\textit{\textbf{Actions:}} In Scenario 1, the channel allocation action is $\mathbf{a}(t)\triangleq\mathbf{D}(t)\in\left\{0,1\right\}^{N\times M}$, under constraint \eqref{Scen1,Constraint}. The discrete action space is denoted as $\mathcal{A}^{(1)}$ with the cardinality of $\left|\mathcal{A}^{(1)}\right|=\frac{N!}{(N-M)!}$ In Scenario 2, both the channel allocation $\mathbf{D}(t)$ and the power control $\mathbf{p}^{\mathrm{tx}}(t)\in\mathbb{P}^N$ are required and thus the hybrid action is $\mathbf{a}(t)\triangleq\left\{\mathbf{D}(t),\mathbf{p}^{\mathrm{tx}}(t)\right\}$. The discrete actions space is denoted as $\mathcal{A}^{(2)}$ with the cardinality of $\left|\mathcal{A}^{(2)}\right|=(M+1)^N$. The hybrid action space is denoted as $\mathcal{A}^{(2)}\times\mathbb{P}^N$.\ In Scenario 3, only power allocation is needed, and the action is $\mathbf{a}(t)\triangleq\mathbf{P}^{\mathrm{tx}}(t)$ belong to the continuous action space $\mathbb{P}^{M\times N}$.

\textit{\textbf{Transitions:}} The transition probability of MDP consists of two components, including the channel state transition of all channels and the AoI state transition of all sensors. Since these two components are independent, the state-transition probability from $\mathbf{s}(t)$ to $\mathbf{s}(t+1)$ under a particular action $\mathbf{a}(t)$ is
\begin{equation}
\operatorname{Pr}[\mathbf{s}(t+1) \mid \mathbf{s}(t), \mathbf{a}(t)]=\operatorname{Pr}[\mathbf{G}(t+1) \mid \mathbf{G}(t)] \operatorname{Pr}[\boldsymbol{\tau}(t+1) \mid \mathbf{s}(t), \mathbf{a}(t)],
\end{equation}
where $\operatorname{Pr}[\mathbf{G}(t+1) \mid \mathbf{G}(t)]$ is the transition probability of the channel state and can be directly obtained from the channel state transmission matrices $\mathbf{M}_n$ in \eqref{ChannelTransition} and $\operatorname{Pr}[\boldsymbol{\tau}(t+1) \mid \mathbf{s}(t), \mathbf{a}(t)]$ is related to the AoI state transition probability.
From the AoI state updating rule~\eqref{AoI,Updating}, it directly follows that
\begin{equation}
\begin{aligned}
&\operatorname{Pr}[\boldsymbol{\tau}(t+1) \mid \mathbf{s}(t), \mathbf{a}(t)] =\prod_{n \in \mathcal{K}(t)} \operatorname{Pr}\left[\tau_{n}(t+1) \mid \tau_{n}(t), \mathbf{G}(t), \mathbf{a}(t)\right], \\
&\operatorname{Pr}\left[\tau_{n}(t+1) \mid \tau_{n}(t), \mathbf{G}(t), \mathbf{a}(t),n \in \mathcal{K}(t)\right] = \begin{cases}\hat{\varepsilon}_{n}(t), & \text { if } \tau_{n}(t+1)=\tau_{n}(t)+1 \\
1-\hat{\varepsilon}_{n}(t), & \text {otherwise.}\end{cases}
\end{aligned}
\end{equation}

\textit{\textbf{Rewards:}} We define the reward as the negative cost function as
$
r(t)=-J(t)=-\sum_{n=1}^{N} J_{n}(t).
$
Thus, the MDP solution aims to find a policy for maximizing the discounted long-term average reward $\lim _{T \rightarrow \infty} \mathbb{E}\left[\sum_{t=0}^{T-1} \beta^{t} r(t)\right]$.

\subsection{Stability Conditions}\label{sec:stability}
Before solving the MDP problems, it is critical to elucidate conditions that the remote estimation system needs to satisfy,  which ensure that there exists at least one stationary and deterministic scheduling policy that can stabilize the remote estimation of all plants.
We derive the necessary and sufficient stability conditions as below.

\begin{proposition}\label{prop:stability}
	\normalfont
(a) Sufficient stability condition. 
For any of the three multiple-access schemes in Section~\ref{sec:MAC}, there always exists a stationary and deterministic resource allocation policy to stabilize the $N$-sensor-$M$-channel remote estimation system in sense of Definition~\ref{def:stability}, if the following holds:
\begin{equation}\label{eq:suf}
\max_{n\in\mathcal{N}} \rho^2_{n} \max_{n\in\mathcal{N}} \lambda_n<1,
\end{equation}
where $\rho_{n} \triangleq \rho(\mathbf{A}_n)$ is the spectral radius of $\mathbf{A}_n$,  and
$
\lambda_n \triangleq \rho\left(\mathbf{\Psi} \mathbf{M}_n\right).
$
$\mathbf{\Psi} \triangleq \operatorname{diag}{\{\psi_1,\psi_2,\dots,\psi_{H^M}\}}$ is a diagonal matrix with $\psi_i \triangleq \varepsilon\left(\max \{\tilde{g}_{i,1},\dots,\tilde{g}_{i,M} \}\frac{P_{\max}}{\sigma^2}\right)$. 
Recall that $\tilde{\mathbf{g}}_i\triangleq (\tilde{g}_{i,1},\dots,\tilde{g}_{i,M})$ denotes the $i$th channel state in the channel state space $\mathcal{G}^M$ defined in Section~\ref{sec:channel}, and $\varepsilon(\cdot)$ is the packet decoding failure probability function.

(b) Necessary stability condition. 
For any of the three multiple-access schemes, there is no stationary and deterministic scheduling policy that can stabilize the remote estimation system, if the following does not hold:
\begin{equation}\label{eq:nec}
\max_{n\in\mathcal{N}} \rho^2_{n}\lambda_n<1,
\end{equation}
\end{proposition}
\begin{proof}
	See Appendix.
\end{proof}

\begin{remark}
The stability conditions depend on the plant system parameters $\{\mathbf{A}_n\}$, the transmission power budget $P_{\max}$, the channel power gain states $\{\tilde{\mathbf{g}}_1,\dots,\tilde{\mathbf{g}}_{H^M}\}$, and  the Markov channel transition rules $\{\mathbf{M}_n\}$.
These conditions show that if the plant systems are of high dynamics (i.e., with large $\rho(\mathbf{A}_n)$), and the transmission power budget is small, the remote estimation system is difficult to be stabilized. 
\end{remark}

\subsection{The Challenges for Solving the MDPs}\label{subsec:MDPchallenge}
\subsubsection{Unknown Channel Parameters and Large State Space}
If the channel parameters, i.e., the state transition matrices $\{\mathbf{M}_n\}$, are available, the state transition probabilities of the formulated MDP problem are also known. Conventional model-based MDP algorithms, including value and policy iteration algorithms, can solve this type of problem with small state and action spaces. For example, a two-sensor-one-channel scheduling problem was solved by the classical value iteration algorithm in \cite{Wu2018SmallAction}. However, in the absence of state transition probabilities, the model-based algorithms cannot work, and thus data-driven reinforcement learning approaches are preferable.

A conventional data-driven reinforcement learning algorithm, such as Q-learning, requires the construction of a Q-value table, occupying a storage space proportional to the state space. In our formulated problem, the size of state space of the formulated MDP problem is $H^{MN}\times\mathbb{N}^N$. It grows exponentially with the increased numbers of sensor and channel, leading to a huge storage space of the Q-value table. To solve the curse of dimensionality introduced by a large state space, more advanced DRL algorithms (e.g., DQN) adopt deep neural networks (DNNs) for value function approximations to replace Q-value tables.

\subsubsection{Large and Hybrid Action Space}
The formulated MDPs in Scenarios 1-3 have large action spaces, even with relatively small numbers of sensors and channels. For example, when $N=10, M=5$, the first two scenarios have $\frac{N!}{\left(N-M\right)!}=30240$ and $\left(M+1\right)^N=60466176$ discrete actions for channel assignment, respectively. For the continuous power allocation action spaces in Scenarios 2 and 3, a commonly adopted method is action discretization \cite{Cao2019ActionDiscretization}. For example, when $N=10, M=5$ and the quantization level is 2, the action space in scenario 2 has $2^N=1024$ discrete actions for power allocation and total $1024\times60466176\approx62\times{10}^9$ discrete actions for both channel assignment and power allocation. The action space in Scenario 3 has $\left(\frac{M!}{\left(M-0\right)!}+\frac{M!}{\left(M-1\right)!}+\frac{M!}{\left(M-2\right)!}\right)^N\approx14\times{10}^{13}$ discrete actions for power allocation.

However, DQN, the most adopted DRL algorithm for discrete action spaces, cannot handle such large action spaces.\footnote{We note that multi-agent reinforcement learning may be able to cope with dimensional complexity, which is beyond the scope of this work. The remote estimator in this work considers a centralized policy by centralized learning, requiring a single-agent reinforcement learning approach.} It is well-known that DQN has been successfully applied in playing various Atari games which have small action spaces (between 4 and 18) \cite{Mnih2013DQN}. Large action space makes the DQN difficult to train and entails large DNNs, requiring large storage space and high computation complexity. Solving MDPs with large action spaces are challenging. Researchers use ad hoc approaches to tackle specific problems, and there are no general solutions.

\subsubsection{Training Difficulties}
In our MDP problems, the channel states are of high dynamics and the AoI state is also a stochastic function of the channel states and the action. In addition, the absolute value of reward grows exponentially in terms of the AoI state, which means a DRL agent needs to deal with a fairly large range of rewards, inducing a highly fluctuated training process that is difficult to converge. This feature is different from many existing works applying DRL in wireless communications problems \cite{Wang2022LogReward,Liu2018LogReward,Chen2021LogReward}, where the reward/cost functions commonly grow up in linear or log scale with the number of states. Therefore, the highly stochastic states and the large reward range are making a DRL agent difficult to converge to an optimal policy.

In the next section, we will develop advanced DRL algorithms with continuous action space and propose novel continuous to discrete action mapping schemes for effectively solving the radio resource allocation MDPs in Scenarios 1-3.

\section{Deep Reinforcement Learning Algorithm} \label{sec:DRL}
There are many model-free DRL frameworks developed in the past decade \cite{Sewak2019Book,Xiao2021DRLpolicy,Moarales2020Book,Mao2022DRLpolicy,Larsen2021TRPO}. The key features of each framework include whether it is value optimization or policy optimization-based; whether it generates a deterministic or stochastic policy; and whether it adopts an on-policy or off-policy learning method. It is critical to match the right DRL framework to the key features of our MDPs. Then, the selected DRL framework can serve as the foundation for effective problem solving, though none of the existing frameworks can solve our problems directly. In the following, we will present a high-level overview of existing DRL frameworks and justify how we choose the right one for our problems.

\subsection{Overview of DRL Algorithms}
\subsubsection{Value-Based vs. Policy-Based DRL}
A value-based DRL algorithm, such as DQN (and its variations) and deep SARSA, can only handle tasks with discrete and small action spaces. Such a framework has a critic NN, approximating the expected long-term future reward given different actions, to evaluate the actions at an input state. The DRL agent applies an $\epsilon$-greedy policy for action selection at each step based on the critic NN outputs. The critic NN is trained for minimizing the approximation errors. After well-training, one can obtain the optimal policy by greedy action searching based on the critic NN. Therefore, value-based DRL obtains the optimal policy indirectly.

Policy-based DRL is a more effective learning approach, which can optimize policy directly. A policy-based DRL algorithm, such as deep deterministic policy gradient (DDPG), twin-delayed DDPG (TD3), soft actor-critic (SAC), trust region policy optimization (TRPO), and proximal policy optimization (PPO), commonly adopt a dual-NN architecture with actor and critic NNs. In particular, an actor NN is introduced for policy approximation and generating action. The actor NN and the critic NN are optimized simultaneously: the actor NN improves its policy by a policy gradient algorithm for achieving the maximum expected long-term future reward approximated by the critic NN; and the critic NN improves its approximation with the same method as in the value-based DRL algorithms.

Furthermore, the actor NN outputs multi-dimensional continuous action directly and can deal with MDPs of continuous action spaces. We will develop new methods for low-dimensional actor NN output design and new schemes for action mapping from continuous actor output to discrete, hybrid discrete and continuous, and continuous radio resource allocation actions in Scenarios 1-3. These will enable policy-based DRL to solve the MDPs with large action space effectively.

\subsubsection{Stochastic vs. Deterministic Policy in DRL}
Policy-based DRL algorithms such as DDPG and TD3 learn deterministic policies, while SAC, TRPO, and PPO learn stochastic policies. As mentioned in Section~\ref{subsec:MDP}, a deterministic policy maps from state to action deterministically as $a=\pi(s)$. A stochastic policy $\pi(a|s)$ is a multi-dimensional probability density function (PDF) determined by the input state \cite{Sewak2019Book}. The DRL agent samples an action at each time based on the PDF. The PDF is usually assumed to be multi-variant Gaussian and parameterized by the actor NN as $\pi(a|s;\theta)$, where $\theta$ denotes the NN parameters. The actor NN outputs the mean and the standard deviation as $\boldsymbol{\mu}(s;\theta)$ and $\boldsymbol{\sigma}(s;\theta)$, respectively.

Existing research shows that stochastic policy based DRL algorithms provide higher speed for training convergence and better training stability than deterministic ones when the training environment is of high dynamics \cite{Xiao2021DRLpolicy}. This is because the former leads to smoother function approximations, which are easier to train \cite{Moarales2020Book}. Thus, it is preferable to consider stochastic policy based DRL for our MDP problems with highly stochastic states. For online deployment, one can directly apply a maximum likelihood scheme for converting the stochastic policy to a deterministic one.

\subsubsection{On-Policy vs. Off-Policy in DRL}
Policy-based DRL algorithms with stochastic policy can be further categorized into on-policy and off-policy learning methods \cite{Mao2022DRLpolicy}. SAC adopts an off-policy learning method, while TRPO and PPO are on-policy ones \cite{Larsen2021TRPO}. An off-policy algorithm learns the optimal policy (approximated by a target NN) that is different from the behavior policy (approximated by a behavior NN) for generating new experiences during training. The past experiences can be stored in a replay buffer to be reused in updating the target network. In contrast, an on-policy learning algorithm updates the policy that is used to generate experiences during training. The generated experiences by the previous policy are discarded after updating.

The off-policy approach has a higher sampling efficiency as it can reuse the past experiences. However, the old experience can be non-beneficial or even harmful to  train the target network, especially in highly stochastic training environments, as the old behavior policy can be entirely different from the optimal policy. Compared with off-policy methods, on-policy ones are more stable in stochastic environments due to the guaranteed monotonic performance improvement based on their policy updating mechanism \cite{Yang2020DeterMDP}. For the existing on-policy DRL algorithms, PPO shows a more simplified computation architecture than TRPO, which is a state-of-the-art on-policy method \cite{Schulman2017PPO}.

Based on the analysis above, we will adopt PPO, a policy-based DRL method that learns a stochastic policy via on-policy learning, for solving the MDPs. In the following, we will present the general framework of PPO first.
Then, we will develop new methods for low-dimensional actor NN output design and new schemes for action mapping from continuous actor output to discrete, hybrid discrete and continuous, and continuous radio resource allocation actions in Scenarios 1-3. These will enable policy-based DRL to solve the MDPs with large action space effectively.

\subsection{PPO-Based DRL Algorithm for Radio Resource Allocation}\label{subsec:PPO}
As shown in Fig.~\ref{fig:DNN}, the PPO agent has a pair of actor and critic neural networks (NNs), where the NN parameters, including weights and biases, are denoted by $\theta$ and $\varphi$, respectively. Each of the NNs has an input layer, multiple hidden layers, an output layer, and adopts a fully connected and feedforward structure. The input layers of the actor and critic NNs receive the MDP state $\mathbf{s}(t)$. All hidden layers adopt the activation function of rectified linear unit (ReLU) for nonlinear function approximation.
\begin{figure}[t]
	\centering\includegraphics[width=6in]{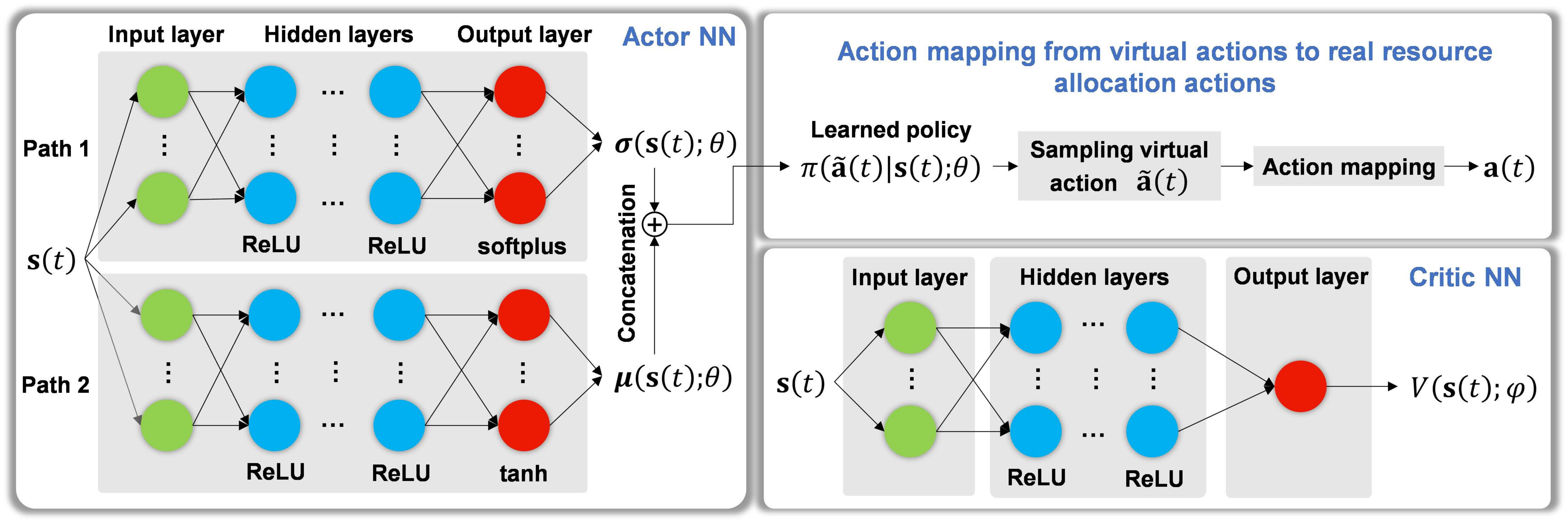}
	\vspace{-0.5cm}
	\caption{An illustration of the PPO's DNN structure with actor and critic NNs, and a action mapping scheme.}
	\label{fig:DNN}
	\vspace{-0.8cm}
\end{figure}

We design the actor NN for generating a multi-dimensional \textit{continuous virtual action} $\tilde{\mathbf{a}}(t)$ with stochastic policy $\pi(\tilde{\mathbf{a}}(t)|\mathbf{s}(t);\theta)$, which is a probability density function of $\tilde{\mathbf{a}}(t)$ given the current state $\mathbf{s}(t)$. The virtual action is then mapped to the real action $\mathbf{a}(t)$ for radio resource allocation as discussed in Section~\ref{sec:ActionMapping} to obtain the next state $\mathbf{s}(t+1)$ and reward $r(t+1)$. The actor output layer generates mean $\boldsymbol{\mu}(\mathbf{s}(t);\theta)$ and standard deviation $\boldsymbol{\sigma}(\mathbf{s}(t);\theta)$ of $\tilde{\mathbf{a}}(t)$. We adopt tanh and softplus as the activation functions for the mean and the standard deviation outputs, respectively. The former is to bound the mean within $[-1,1]$, while the latter is to guarantee a positive standard deviation. The critic NN estimates the state-value function $V(\mathbf{s}(t);\varphi)$ given the actor’s policy $\pi(\tilde{\mathbf{a}}(t)|\mathbf{s}(t);\theta)$, i.e.,
\begin{equation}
V(\mathbf{s}(t) ; \varphi)\approx\mathbb{E}\left[\sum_{k=0}^{\infty} \beta^{k} r(t+k) \bigg| \pi(\tilde{\mathbf{a}}(t+k) | \mathbf{s}(t+k) ; \theta), \forall k \geq 0\right].
\end{equation}
Thus, the critic NN has a single output.

The training of the PPO agent alternates between experience generation and policy update. In the following, we only present the key steps and the high-level concepts of PPO. The detailed explanation can be found in~\cite{Schulman2017PPO}.
The training process is presented in Algorithm~\ref{alg:training}.
\begin{algorithm}[t]
\small
\caption{\small{Training Algorithm of the DRL-based Radio Resource Allocation for Remote State Estimation}}\label{alg:training}
    \begin{algorithmic} [1]
    \Require Episode number $E$, maximum steps per episode $T$, number of epochs for NN updating $K$, discount factor $\gamma$, smoothing factor $\alpha$ of generalized advantage estimator, mini-batch size $B$, clip factor $\omega$, the entropy loss weight factor $w$, the size of the actor output layer $C$.
    \Ensure Well-trained actor network $\pi^*(\cdot|\cdot)$.
    \State Initialize actor network $\pi(\cdot|\cdot)$ and critic network $V(\cdot)$ with random parameter $\theta$ and $\varphi$,respectively; Initialize the AoI of all sensors to zeros; Randomly initialize the channel states of all sensors.
    \For{episode =1,…,$E$}
        \For{t=0,…,$T$}
        \State Generate experiences $<\mathbf{s}(t),\tilde{\mathbf{a}}(t),r(t)>$ by following the current policy $\pi(\cdot|\cdot)$.\label{alg:experiences}
        \EndFor
        \For{t=0,…,$T$}
        \State Compute the advantage function $A(t)$ in \eqref{AdvantageFunc} and the reward-to-go $R(t)$ in \eqref{reward-to-go}.
        \EndFor
        \For{epoch =1,…,$K$}
        \State Sample a random mini-batch data set of size $B$ from the experiences.
        \State Update the critic parameters $\varphi$ by minimizing the loss function $L_{\text{C}}(\varphi)$ in \eqref{eq:LossForCritic}.
        \State Update the actor parameters $\theta$ by minimizing the actor loss function $L_{\text{A}}(\theta)$ in \eqref{eq:LossForActor}.
        \EndFor
    \EndFor
    \end{algorithmic}
\end{algorithm}

1) Experience generation. 
By using current policy $\pi(\cdot|\cdot;\theta_{\text{old}})$, the PPO agent samples data $(\mathbf{s}(t),\tilde{\mathbf{a}}(t),r(t))$ with length of $T$ through interacting with the environment [see Algorithm~\ref{alg:training} line~\ref{alg:experiences}].
By leveraging the generated experience, the advantage function $A(t)$ and the reward-to-go function $R(t)$ for each $t=0,\dots,T-1$ can be calculated as
\begin{align}\label{AdvantageFunc}
&A(t)=\sum_{k=t}^{T-1}(\beta \alpha)^{k-t} (r(k)+\beta V(\mathbf{s}(k+1) ; \varphi)-V(\mathbf{s}(k) ; \varphi)),\\  \label{reward-to-go}  
&R(t)= r(t) + \beta V(\mathbf{s}(t+1) ; \varphi),
\end{align}
where $\alpha$ is a hyper-parameter named as the smoothing factor.

2) Policy update.
The PPO agent creates a mini-batch data set by randomly sampling $B$ data from the generated $T$-length experience earlier, i.e., $\{(\mathbf{s}(t_i),\tilde{\mathbf{a}}(t_i), A(t_i), R(t_i))\}$ where $t_i \in \{t_1,\dots,t_B\} \subset \{0,\dots,T-1\}$. The loss function for updating the critic NN is defined as
\begin{equation}\label{eq:LossForCritic}
L_{\text {C}}(\varphi)=\frac{1}{B} \sum_{i=1}^{B}\left(R(t_i)-V\left(\mathbf{s}(t_i) ; \varphi\right)\right)^{2}.
\end{equation}
This is regarded as the temporal difference error, specifying difference of the state-value estimations based on time steps $t_i$ and $t_i+1$.The loss function for updating the actor NN has been elaborately designed for achieving high training stability and is much more complex:
\begin{equation}\label{eq:LossForActor}
L_{\text {A}}(\theta)=\frac{1}{B} \sum_{i=1}^{B}\left(\min\left\{p(\mathbf{s}(t_i);\theta)A(t_i), c\left(\mathbf{s}(t_i) ; \theta\right)A(t_i)\right\}+w \hbar_{\theta_{\text {old}}}\left(\tilde{\mathbf{a}}(t_i) \right)\right),
\end{equation}
where $
c(\mathbf{s}(t_i);\theta)\!=\!\max\!\left\{\min\! \left\{p(\mathbf{s}(t_i);\theta), 1\!+\!\omega\right\}, 1\!-\!\omega\right\}
$ is a clip function with a hyper-parameter $\omega$, and 
$p(\mathbf{s}(t_i);\theta) = \frac{\pi(\tilde{\mathbf{a}}(t_i)|\mathbf{s}(t_i) ; \theta)}{\pi\left(\tilde{\mathbf{a}}(t_i)| \mathbf{s}(t_i) ; \theta_{\text {old }}\right)}$ is a density ratio. 
$\hbar_{\theta_{\text {old}}}\!\!\left(\tilde{\mathbf{a}}(t_i) \right)$ is the entropy loss function of $\tilde{\mathbf{a}}(t_i)$. Since $\tilde{\mathbf{a}}(t_i)$ follows a Gaussian distribution, its entropy loss can be directly obtained based on its standard deviation $\boldsymbol{\sigma}(\mathbf{s}(t_i);\theta_{\text{old}})$.
$w$ denotes entropy loss weight factor. 
Then, the critic NN and the actor NN parameters $\varphi$ and $\theta$ can be updated by optimizing the loss functions $L_{\text {C}}(\varphi)$ and $L_{\text {A}}(\theta)$, respectively, using the widely adopted Adam optimizer.

After training, the virtual action $\tilde{\mathbf{a}}(t)$ can be generated deterministically based on the maximum likelihood method for online deployment, i.e., $\tilde{\mathbf{a}}(t) = \boldsymbol{\mu}(\mathbf{s}(t);\theta)$.

\subsection{Actor NN Output Design and Action Mapping Schemes} \label{sec:ActionMapping}
We will present the schemes for low dimensional actor NN output design and continuous virtual action to continuous, hybrid, and discrete real action in the sequel in Scenarios 1-3. The schemes are illustrated in Fig.~\ref{fig:ActionMapping}.
\begin{figure}[t]
	\centering\includegraphics[width=\textwidth]{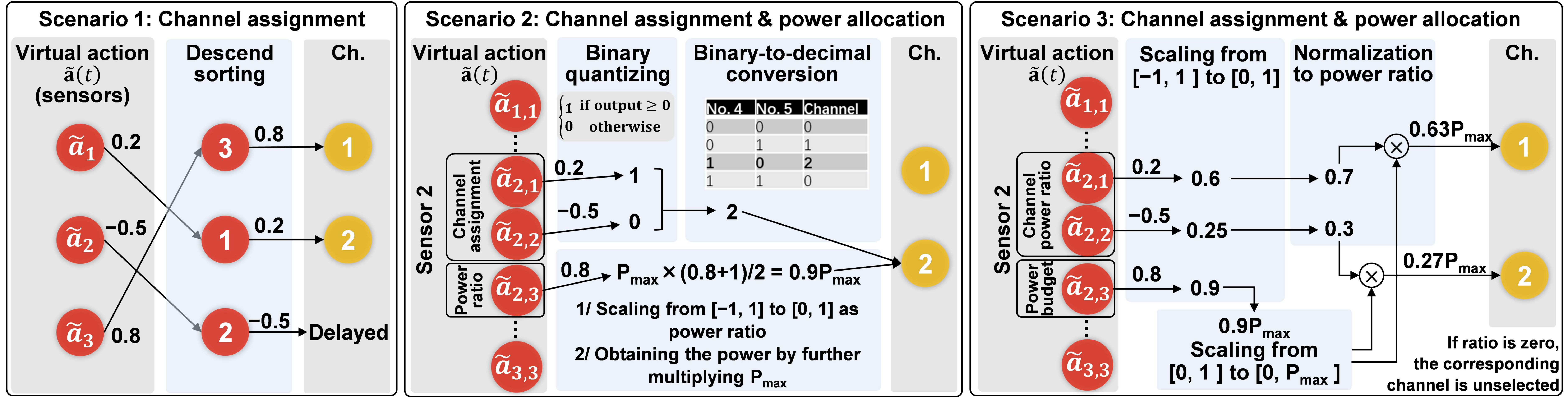}
	\vspace{-1cm}
	\caption{Action mapping schemes for DNN models of Scenarios 1-3, illustrated with a system consisting of three sensors and two frequency channels.}
	\label{fig:ActionMapping}
	\vspace{-0.8cm}
\end{figure}

\subsubsection{Scenario 3}
The power allocation MDP has continuous actions. We design the virtual action of each sensor with $M+1$ dimensions as $[\tilde{a}_{n,1},\dots, \tilde{a}_{n,M+1}]$, where $\tilde{a}_{n,m}\in \mathbb{R}$. The virtual action $\tilde{a}_{n,M+1}$ indicates sensor $n$’s total power ratio to be allocated in $M$ channels and $\tilde{a}_{n,m}$ represents the allocation ratio at channel $m$. Specifically, the total transmit power is obtained as
\begin{equation}
    P^{\mathrm{tx}}_n =\frac{\operatorname{clamp_{-1}^1}[\tilde{a}_{n,M+1}]+1}{2} P_{\text{max}},
\end{equation}
where the operator $\operatorname{clamp_{-1}^1}[\cdot]$ is for truncating a variable to values between $-1$ and $1$.
Then, based on the normalization, sensor $n$’s transmit power at channel $m$ is
\begin{equation}
    P^{\mathrm{tx}}_{n,m} = \frac{(\operatorname{clamp_{-1}^1}[\tilde{a}_{n,m}]+1)/2}{\sum_{m=1}^{M} \left(\left(\operatorname{clamp_{-1}^1}[\tilde{a}_{n,m}]+1\right)/2\right)}P^{\mathrm{tx}}_n.
\end{equation}
It is easy to see that the virtual action has $N(M+1)$ actions in total, and thus the actor NN has $2N(M+1)$ outputs, i.e., a linear function of either $N$ or $M$.

\subsubsection{Scenario 2} \label{sec:ActionMapping2}
Considering the hybrid action feature of Scenario 2, it is natural to design the continuous virtual action for sensor $n$ into two parts for channel selection and power allocation. However, how to design a continuous virtual action and map it to a discrete channel selection action for effective DRL is highly non-trivial. An easy-to-think method is to adopt a scaler continuous virtual action for each sensor and discrete it linearly into $M+1$ levels for channel selection. Level 0 means no channel selection and Level $m$ denotes the selection of channel $m$. At the first glance, the method compresses the size-$M$ discrete action space by a one-dimensional continuous one, which is space reducing. However, the method implicitly converts the original non-Euclidean action space into a Euclidean one. This potentially requires the actor NN to approximate a highly discontinuous multi-level-multi-stage piecewise function. Ideally, a slight change of state can make a significant change in the actor NN output \cite{Moarales2020Book}. This makes the PPO policy difficult to converge to an optimal one.

To solve the issue of discontinuous function approximation, we propose to convert the decimal channel selection action (i.e., from $0$ to $M$) of each sensor to a binary sequence of length $\lceil \log_{2}(M+1) \rceil$. Thus, we design the virtual action for each sensor with $\lceil \log_{2}(M+1) \rceil $ outputs and the positive and negative entries are mapped to ‘1’ and ‘0’, respectively, for real action mapping. Now, there are $N \lceil \log_{2}(M+1) \rceil$ actions in total for channel selection of all sensors. Each of these virtual action entries only handles the selection of two discrete actions. Changing of a single virtual action entry from negative to positive represents the reduced and the increased likelihood of one action and the other, respectively. Thus, the actor NN is continuous in each output dimension, and is easy for training. The original size-$(M+1)^N$ discrete action space is compressed by the $N \lceil \log_{2}(M+1) \rceil $-dimension continuous one, which scales linearly and logarithmically in $N$ and $M$, respectively.

Same as Scenario 3, one-dimensional virtual action output $\tilde{a}_n \in \mathbb{R}$ is added for each sensor’s transmit power control. The mapping between the virtual to real power control action is
$P^{\mathrm{tx}}_n =P_{\max} (\operatorname{clamp_{-1}^1}[\tilde{a}_n]+1)/2$. Thus, the actor NN has $2N \lceil \log_{2}(M+1) \rceil $ outputs in total.

\subsubsection{Scenario 1}
Compared to Scenario 2, Scenario 1 has a stricter constraint on the channel selection actions, i.e., different sensors cannot choose the same channel. Thus, the virtual action design method for Scenario 2, which assigns virtual action outputs for different sensors separately, cannot be applied to Scenario 1, which requires joint representation of all sensor actions.

We design the continuous virtual action with $N$ dimensions, where the entries are associated with sensor 1 to $N$, respectively. For the virtual action to real action mapping, we sort the virtual actions in descending order. The sensors associated with the largest $M$ virtual action entries are assigned to channels $1$ to $M$, sequentially. Therefore, we compress the original size-$N!/(N-M)!$ discrete action space with the $N$-dimension continuous one. The actor NN has $2N$ outputs.

\section{Numerical Experiments} \label{sec:simulation}
In the following, we evaluate and compare the performance of the DRL-based resource allocation algorithms of Section~\ref{sec:DRL}.

\subsection{Experiment Setups}
Our numerical experiments are implemented using  MATLAB 2021b based on the computing platform with two Intel Xeon Gold 6256 CPUs @ 3.60 GHz and a 192 GB RAM. GPU is not required. 
Each of the actor and critic NNs of the PPO-based DRL agent has three hidden layers with sizes of $\left\lceil70K\right\rceil, \left\lceil50K\right\rceil,\left\lceil30K\right\rceil$, respectively, where $K=\sqrt{N/M}\log _{2}(M+1)$. The state input of each NN has $N(M+1)$ dimensions, which is the same as the MDP. The output size of the critic NN is $1$, while that of the actor NN is discussed in Section~\ref{sec:ActionMapping}. The dynamic system matrices $\mathbf{A}_n$ are randomly generated by leveraging the method presented in \cite{Leong2020OMA}, where the spectrum radius is drawn uniformly from the range of $(1, 1.3)$. 
The channel transition matrices are generated randomly.
Only the system matrices and the channel transition matrices that satisfy the sufficient stability condition established in Section~\ref{sec:stability} are adopted in the experiment.
The details of the remote estimation system parameters and the DRL parameters are summarized in Table~\ref{tab:Setup}.

\begin{table}[t]
\footnotesize
\setlength\tabcolsep{0.5pt}
\centering
\caption{Summary of Experimental Setup}
\vspace{-0.5cm}
\begin{tabular}{cc|cc}
\hline\hline
Items   & Value     & Items     & Value   \\ \hline
\multicolumn{2}{l|}{\textbf{Remote state estimation system parameters}} & \multicolumn{2}{l}{\textbf{Agent parameters of the benchmark (DQN)}} \\
\rowcolor[HTML]{EFEFEF} 
Transmit power budget {[}dBm{]}, $P_{\max}$    & 23      & Initial epsilon for exploring action space    & 1        \\
Background noise power {[}dBm{]}, $\sigma^2$    & $-$60     & Epsilon decay rate      & 0.999 \\
\rowcolor[HTML]{EFEFEF} 
Code rate {[}bps{]}, $b/l$    & 2        & Minimum epsilon  & 0.01      \\
Block length {[}symbols{]}, $l$    & 200     & Mini-batch size, $B$   & 128   \\
\rowcolor[HTML]{EFEFEF} 
Channel states for Scenario 1 \& 2, $\mathcal{G}$  &  $\{10^{-8},10^{-7},\cdots,10^{-1}\}$  & Experience buffer length   & $1000NM$    \\
Channel states for Scenario 3, $\mathcal{G}$   & $\{\sqrt{10^{-8}},\sqrt{10^{-7}},\cdots,\sqrt{10^{-1}}\}$  & Target critic network update frequency  & 24      \\
\rowcolor[HTML]{EFEFEF} 
Channel state transition matrix  & Randomly generalized & Discount factor, $\beta$     & 0.95    \\\hline
\multicolumn{2}{l|}{\textbf{Traning parameters}}                        & \multicolumn{2}{l}{\textbf{Agent parameters of this work (PPO)}}      \\
\rowcolor[HTML]{EFEFEF} 
Episode number, $E$    &  $\left\lceil 250 \times \frac{N}{M} \times \sqrt{NM} \right\rceil$   & Number of epochs for learning, $K$     & 3     \\
Maximum time steps per episode, $T$     & 128      & Generalized advantage estimator factor, $\alpha$   & 0.95    \\
\rowcolor[HTML]{EFEFEF} 
Learning rate of actor network (PPO)    & 0.0001   & Mini-batch size, $B$    & 128      \\
Learning rate of critic network     & 0.001    & Clip factor, $\omega$   & 0.2      \\
\rowcolor[HTML]{EFEFEF} 
Optimizer during learning     & Adam      & Discount factor, $\beta$      & 0.95       \\
Threshold value of the learning gradient & 1 & Entropy loss weight factor, $w$  & 0.01   \\ \hline\hline
\end{tabular}
\label{tab:Setup}
\vspace{-0.5cm}
\end{table}

We adopt the DQN algorithm as the benchmark policy of Scenario~1 with the OMA scheme. 
As shown in~\cite{Leong2020OMA}, the DQN-based algorithm performs better than the heuristic algorithms (e.g., the greedy and the round-robin policies). Thus, we only need to compare our algorithm with the DQN-based one. 
For Scenario~2, we use the naive action mapping scheme-based PPO (discussed in Section~\ref{sec:ActionMapping}) as the benchmark.
The DRL-based resource allocation algorithm in Scenario~3 is then compared with Scenarios~1 and~2.

\subsection{Performance Evaluation of the DRL Algorithms}
In Table~\ref{tab:scalability}, we compare the remote estimation performance, i.e., the average sum MSE of all plants, among the proposed DRL-based algorithms in Scenarios 1-3 and the benchmark policies with various system scales and sensor-to-channel ratios (SCRs).
The remote estimation system with a smaller average estimation MSE has a better  performance.
Average estimation MSEs in Table~\ref{tab:scalability} are calculated by $10000$-step simulations.
\begin{table}[t]
\begin{minipage}{.6\linewidth}
    \footnotesize
    \setlength\tabcolsep{3pt}
    \centering
    \caption{Performance Comparison of the Proposed Algorithm and the Benchmarks in Terms of Averaged Estimation MSE}
    \vspace{-0.5cm}
    \label{tab:scalability}
\begin{tabular}{c|cc|cc|c}
\hline\hline
System scale & \multicolumn{2}{c}{Scenario 1} \vline & \multicolumn{2}{c}{Scenario 2} \vline & Scenario 3 \\ \hline
($N,M$, SCR) & DQN & This work & \thead{Naive action \\ mapping} & This work & This work \\ \hline
\rowcolor[HTML]{EFEFEF} 
(6, 3, 2) & 46.6243 & 44.7455 & 39.0462 & 38.0663 & 45.8611 \\
(10, 5, 2) & $-$ & 73.3188 & 65.0766 & 63.0768 & 50.5058 \\
\rowcolor[HTML]{EFEFEF} 
(20, 10, 2) & $-$ & 158.4659 & 154.9627 & 125.4118 & 90.5851 \\
(30, 15, 2) & $-$ & 269.2111 & 297.5236 & 218.3914 & 134.0544 \\
\rowcolor[HTML]{EFEFEF} 
(40, 20, 2) & $-$ & 403.6164 & 397.5561 & 286.1228 & 175.6329 \\
(50, 25, 2) & $-$ & 525.6877 & 486.9513 & 360.9214 & 238.1790 \\ \hline
\rowcolor[HTML]{EFEFEF} 
(10, 4, 2.5) & $-$ & 83.3288 & 78.8600 & 77.6075 & 64.2634 \\
(20, 8, 2.5) & $-$ & 196.6553 & 172.9613 & 157.2169 & 94.2942 \\
\rowcolor[HTML]{EFEFEF} 
(30, 12, 2.5) & $-$ & 69.2031 & 318.8591 & 252.8929 & 141.1223 \\
(40, 16, 2.5) & $-$ & 568.8016 & 404.8909 & 345.0570 & 178.3883 \\
\rowcolor[HTML]{EFEFEF} 
(50, 20, 2.5) & $-$ & 35.9785 & 563.4724 & 434.4692 & 239.0142 \\ \hline
(15, 5, 3) & $-$ & 183.0668 & 159.5629 & 129.8431 & 82.3815 \\
\rowcolor[HTML]{EFEFEF} 
(24, 8, 3) & $-$ & 304.3899 & 275.5904 & 230.7575 & 111.8687 \\
(33, 11, 3) & $-$ & 524.6109 & 438.7793 & 343.1833 & 157.6327 \\
\rowcolor[HTML]{EFEFEF} 
(42, 14, 3) & $-$ & 731.5286 & 578.4324 & 448.6448 & 193.3498 \\
(51, 17, 3) & $-$ & 1066.5461 & 765.4839 & 592.1996 & 246.8475 \\ \hline\hline
\end{tabular}
\end{minipage}\hfill
\begin{minipage}{.4\linewidth}
    \footnotesize
    \setlength\tabcolsep{1pt}
    \centering

    \caption{Computation Time of the Proposed Algorithms for Decision Making in a Resource Allocation Cycle (Unit: ms)}
    \label{tab:computation}
    \vspace{-0.5cm}
\begin{tabular}{cccc}
\hline\hline
\thead{System scale \\ ($N,M$, SCR)} &Scenario 1 & Scenario 2 & Scenario 3 \\ \hline
\rowcolor[HTML]{EFEFEF} 
(6, 3, 2)  & 0.8665 & 0.8889 & 0.8155 \\
(10, 5, 2) & 0.8660 & 0.8977 & 0.8685 \\
\rowcolor[HTML]{EFEFEF} 
(20, 10, 2) & 0.9633 & 1.0232 & 1.1534 \\
(30, 15, 2) & 1.0773 & 1.2392 & 1.7283 \\
\rowcolor[HTML]{EFEFEF} 
(40, 20, 2) & 1.2935 & 1.4748 & 1.7087 \\
(50, 25, 2) & 1.6462 & 1.8995 & 2.0258 \\ \hline
\rowcolor[HTML]{EFEFEF} 
(10, 4, 2.5) & 0.8521 & 0.9293 & 0.9190 \\
(20, 8, 2.5) & 1.0042 & 1.0821 & 1.0263 \\
\rowcolor[HTML]{EFEFEF} 
(30, 12, 2.5) & 1.1423 & 1.2412 & 1.3685 \\
(40, 16, 2.5) & 1.3768 & 1.4548 & 1.7103 \\
\rowcolor[HTML]{EFEFEF} 
(50, 20, 2.5) & 1.5651 & 1.8333 & 2.0271 \\ \hline
(15, 5, 3) & 0.9099 & 0.9721 & 0.8769 \\
\rowcolor[HTML]{EFEFEF} 
(24, 8, 3) & 1.0567 & 1.0564 & 0.9209 \\
(33, 11, 3) & 1.2135 & 1.3257 & 1.3358 \\
\rowcolor[HTML]{EFEFEF} 
(42, 14, 3) & 1.4302 & 1.5020 & 1.4890 \\
(51, 17, 3) & 1.5488 & 1.9323 & 2.1405 \\ \hline\hline
\end{tabular}
\end{minipage}
\vspace{-0.8cm}
\end{table}

We see that the DQN algorithm only works for the 6-sensor-3-channel setting, while all the proposed DRL algorithms can scale up to 50 sensors and 25 channels.
In Scenario 2 (i.e., the conventional NOMA scheme), we see that the proposed action mapping scheme-based algorithm can provide a 25\% average estimation error reduction than the naive mapping scheme-based one.

It can be observed that the optimized resource allocation policy with the proposed NOMA scheme (Scenario 3) outperforms the conventional NOMA scheme, which performs better than the OMA scheme (Scenario 1) in general. The performance gains of the proposed NOMA scheme over the conventional NOMA increase significantly with the increased system scale. In particular, the proposed NOMA scheme can reduce the average estimation error by 50\% in some large settings, e.g., $N=51$ and $M=17$.
In Fig.~\ref{fig:TriScene}, we plot the training curves of the proposed DRL algorithms in Scenarios 1-3 with $N=51$ and $M=17$.
It is interesting to see that Scenario 3 not only provides the best training performance but also requires the shortest time for convergence among Scenarios 1-3.
This indicates that the proposed NOMA scheme is preferable to the conventional OMA and NOMA schemes, especially in large-scale remote estimation systems.

In Table~\ref{tab:scalability}, we see that the estimation performance decreases with the increasing SCR as expected, due to the decreasing wireless resources. In Fig.~\ref{fig:TriSCR}, we also plot the training processes with different SCR in Scenario 1. It can be observed that a larger SCR induces a less stable training process. 

\begin{figure}[t]
	\centering
	\subfigure[]{
		\begin{minipage}[b]{0.47\textwidth}
			\includegraphics[width=1\textwidth]{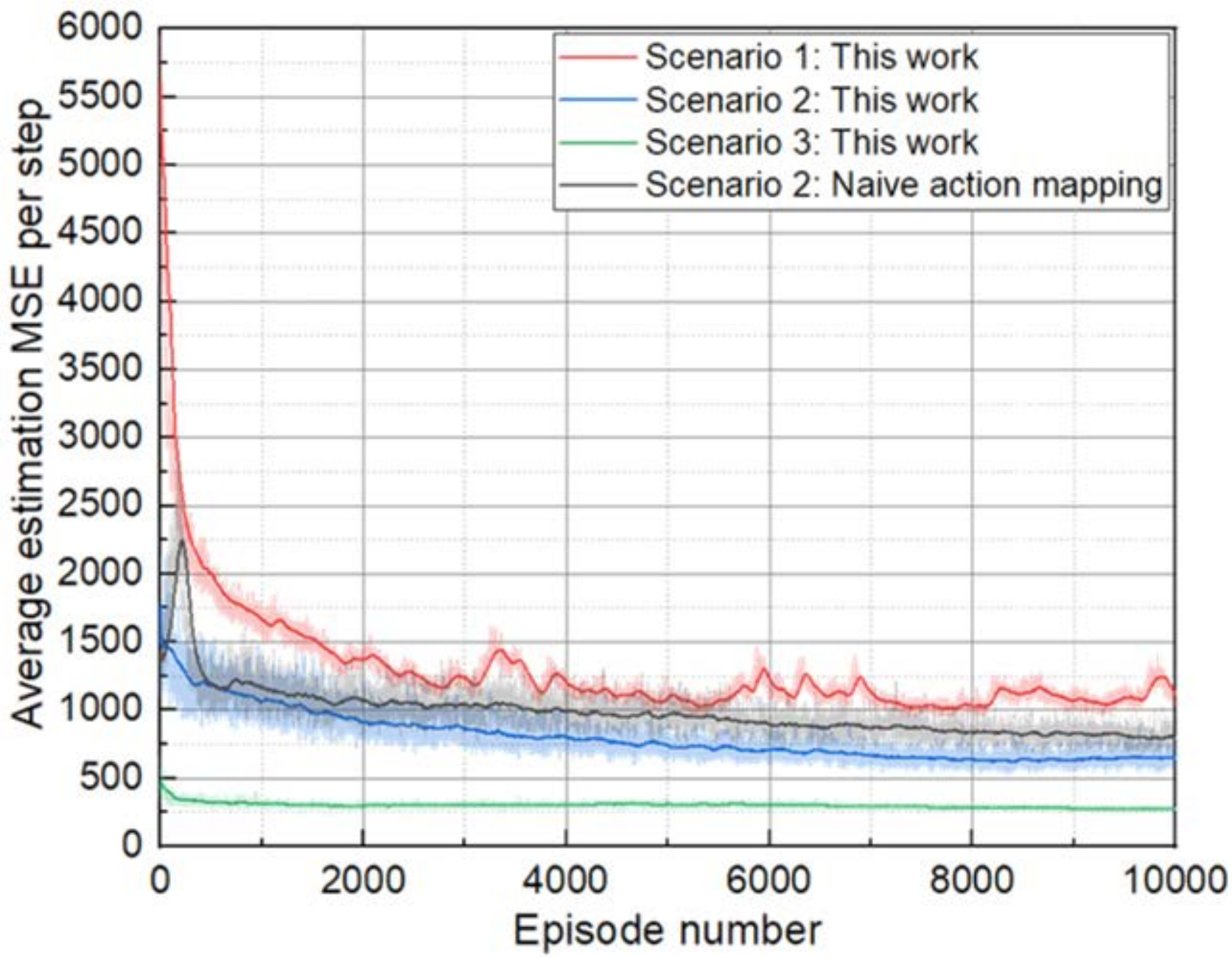}
		\end{minipage}
		\label{fig:TriScene}
	}
    	\subfigure[]{
    		\begin{minipage}[b]{0.47\textwidth}
   		 	\includegraphics[width=1\textwidth]{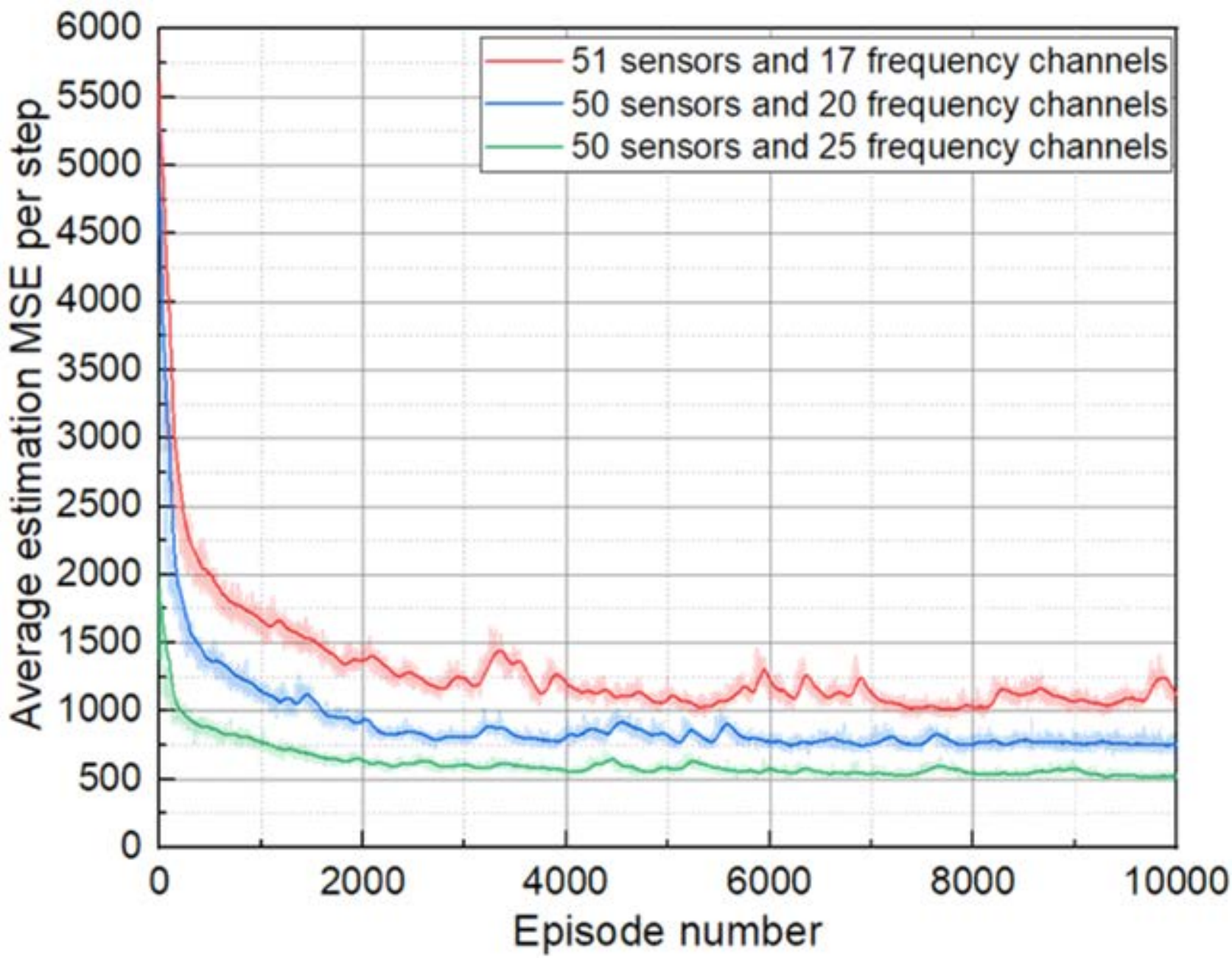}
    		\end{minipage}
		\label{fig:TriSCR}
    	}
    \vspace{-0.5cm}
	\caption{Traning process of proposed PPO-based DRL. (a) The training process of the proposed PPO-based DRL methods for Scenarios 1-3 with a system consisting of 51 sensors and 17 frequency channels. (b) The training process of the proposed PPO-based DRL methods for Scenario 1 with the systems with different SCRs.}
	\label{fig:TraningCurve}
	\vspace{-0.8cm}
\end{figure}

\subsection{Computation Complexity for Online Deployment}
For online deployment, only the well-trained actor NNs are adopted for resource allocation.
In Table~\ref{tab:computation}, we compare the average computation time of the NNs obtained by the proposed DRL-based algorithms in each resource allocation cycle.
We see that the computation time grows up with the increasing system scale in each scenario. 
In particular, it can be observed that the actor NNs in Scenario 3 and 1 have the highest and the lowest computation complexity, respectively. Together with Table~\ref{tab:scalability}, it is interesting to see that the proposed NOMA-based implementation can provide a 50\% performance improvement while consuming only 10\% more computation resources than the conventional NOMA-based implementation in large remote estimation systems (e.g., $N=51$ and $M =17$).

\section{Conclusions}\label{sec:conclusion}
We have proposed practical remote estimation systems over the conventional OMA and NOMA schemes and the proposed NOMA scheme. We have derived the stability conditions of the systems and developed advanced DRL algorithms for resource allocation with large state and action spaces. Our experiments have shown that the proposed algorithms can effectively solve the resource allocation problems. Furthermore, the optimized resource allocation policy with the proposed NOMA scheme has significantly improved estimation performance compared to the conventional OMA and NOMA scenarios. For future work, we will investigate distributed DRL algorithms for resource allocation of large-scale systems and compare them with the present centralized allocation schemes.

\begin{appendices}
\renewcommand{\thesectiondis}[2]{\!\!:}
\section{Proof of Proposition~\ref{prop:stability}}

Due to the space limitation, we only present a brief proof building on our previous works~\cite{Liu2022Stability1,Liu2021FSMC}.

(a) To prove the sufficiency, we construct a
scheduling policy with the OMA scheme and then derive the sufficient condition under which the remote estimator is stabilized by the policy. 
We note that since the NOMA schemes performs better performance than OMA, there exists channel allocation and power allocation policies to stabilize the remote estimator with NOMA if the system can be stabilized with OMA.

For tractability, we construct a persistent scheduling policy, which persistently schedules sensor $1$'s transmission until it is successful and then schedules sensor $2$ and so on.
We adopt the estimation cycle-based analytical approach we developed earlier~\cite{Liu2022Stability1},
where the infinite time horizon is divided by the sensors' successful transmission events as remote estimation cycles.
Then, the average estimation MSE in the infinite time horizon can be converted as the average sum MSE per cycle. 
By following the similar analytical steps, we first derive an upper bound of the average sum MSE per cycle, and then we can show that the average sum MSE per cycle is bounded once \eqref{eq:suf} holds.

(b) To prove the necessity, we have to construct a
virtual resource allocation policy, which always performs better
than any real policy at any multiple-access scheme.
If we can find the necessary condition making the remote estimator stable under the virtual policy, the condition is still necessary for any real policy that stabilizes the system.

For tractability, it is convenient to consider a virtual system, where only one of the sensors needs to be
scheduled for transmission and all other sensors always have perfect transmissions without using any
of the channels (i.e., with zero communication
overhead). Thus, the selected sensor is scheduled
for transmission at each time and uses the best channel with the highest channel power gain. In this case, since only one
of the $N$ sensors has packet detection failures, the remote estimation of all other plants are stabilized.
Then, the necessary stability
condition of the single-plant remote estimation system can be obtained based on our previous work~\cite{Liu2021FSMC}. Therefore, it is easy to derive the overall necessary stability condition~\eqref{eq:nec} considering the selection of the $N$ different sensors.

\end{appendices}

    \balance
    
	\ifCLASSOPTIONcaptionsoff
	\newpage
	\fi


\end{document}